\documentclass[a4paper]{IEEEtran}

\usepackage{amsmath}
\usepackage{amssymb}
\usepackage{bm}

\usepackage[english]{babel}

\usepackage{graphicx}
\usepackage{psfrag}

\usepackage{tikz}
\usetikzlibrary{arrows,backgrounds}

\usepackage{amsthm}

\theoremstyle{plain}
\newtheorem{theorem}{Theorem}
\newtheorem{lemma}[theorem]{Lemma}

\newtheorem{proposition}[theorem]{Proposition}

\newtheorem{definition}{Definition}
\newtheorem{assumption}{Assumption}

\theoremstyle{remark}
\newtheorem{remarkenv}{Remark}
\newenvironment{remark}{\begin{remarkenv}}{\hfill\raisebox{0.2mm}[0cm][0cm]{\scalebox{1.1}{$\lhd$}}\end{remarkenv}}
\newtheorem{exampleenv}{Example}

\newcommand{\sortbib}[1]{}

\usepackage{dsfont}
\newcommand{\N}{\ensuremath{\mathds{N}}} 
\newcommand{\R}{\ensuremath{\mathds{R}}} 

\newcommand{\dd}{\ensuremath{\mathrm{d}}}                         
\newcommand{\di}{\ensuremath{\:\mathrm{d}}}                       
\newcommand{\defl}{\ensuremath{\mathrel{\mathop:}=}}              
\newcommand{\defr}{\ensuremath{=\mathrel{\mathop:}}}              
\newcommand{\T}{\ensuremath{\mathrm{T}}}                          

\usepackage{xspace}
\newcommand{\classK}{\ensuremath{\mathcal{K}}\xspace}             
\newcommand{\classKinf}{\ensuremath{\mathcal{K}_{\infty}}\xspace} 
\newcommand{\classKL}{\ensuremath{\mathcal{KL}}\xspace}           

\newcommand{\Linf}{\ensuremath{\mathcal{L}_{\infty}}}

\newcommand{\Ical}{\ensuremath{\mathcal{I}}}

\newcommand{\Lcal}{\ensuremath{\mathcal{L}}}

\newcommand{\Scal}{\ensuremath{\mathcal{S}}}
\newcommand{\Tcal}{\ensuremath{\mathcal{T}}}

\newcommand{\Xcal}{\ensuremath{\mathcal{X}}}

\newcommand{\smin}{\ensuremath{\text{\rm min}}}
\newcommand{\smax}{\ensuremath{\text{\rm max}}}
\newcommand{\sref}{\ensuremath{\text{\rm ref}}}

\newcommand{\infnorm}[2]{%
    \ensuremath{\|#1\|_{\infty\vphantom{l}}^{\raisebox{0pt}[6.5pt][0pt]{$\scriptstyle #2$}}}}
\newcommand{\id}{\ensuremath{\text{id}}}
\newcommand{\dt}{\ensuremath{\Delta t}}
\newcommand{\vref}{\ensuremath{v_{\sref}}}

\title{\LARGE String stability and a delay-based spacing policy for vehicle platoons subject to disturbances}
\author{B.~Besselink, K.H.~Johansson%
\thanks{This research is financially supported by the European Union Seventh Framework Programme under the project COMPANION, the Swedish Research Council, and the Knut and Alice Wallenberg Foundation.\newline B.~Besselink is with the Johann Bernoulli Institute for Mathematics and Computer Science, University of Groningen, Groningen, the Netherlands (email: b.besselink@rug.nl). K.H.~Johansson is with the ACCESS Linnaeus Centre and the Department of Automatic Control, School of Electrical Engineering, KTH Royal Institute of Technology, Stockholm, Sweden (email: kallej@kth.se). The research reported in this work was performed when the first author was at KTH Royal Institute of Technology, Stockholm, Sweden.}}

\begin{document}

\maketitle

\begin{abstract} 
A novel delay-based spacing policy for the control of vehicle platoons is introduced together with a notion of disturbance string stability. The delay-based spacing policy specifies the desired inter-vehicular distance between vehicles and guarantees that all vehicles track the same spatially varying reference velocity profile, as is for example required for heavy-duty vehicles driving over hilly terrain. Disturbance string stability is a notion of string stability of vehicle platoons subject to external disturbances on all vehicles that guarantees that perturbations do not grow unbounded as they propagate through the platoon. Specifically, a control design approach in the spatial domain is presented that achieves tracking of the desired spacing policy and guarantees disturbance string stability with respect to a spatially varying reference velocity. The results are illustrated by means of simulations.
\end{abstract}

\section{Introduction}
Intelligent transportation systems have the potential to increase efficiency and safety of road transportation through the use of increased automation. Platooning, which amounts to the operation of vehicles in closely-spaced groups, is a particularly relevant example in which the reduced distances between vehicles lead to a decreased aerodynamic drag and a better utilization of the road infrastructure \cite{varaiya_1993}. In particular, experiments using heavy-duty vehicles have shown that the improved aerodynamics of the group leads to a reduction of fuel consumption of up to ten percent, see~\cite{alam_2015b} and~\cite{bonnet_2000}.

In order to ensure the safe operation of such platoons of closely-spaced vehicles, automation of the longitudinal dynamics is required. Early works on this topic are given by \cite{levine_1966} and \cite{chu_1974} and many results have appeared since, see, e.g., \cite{stankovic_2000,jovanovic_2005,barooah_2009,naus_2010,zhang_2015}. Two fundamental aspects in the resulting behavior are, firstly, the spacing policy, which specifies the desired (and not necessarily static) inter-vehicular distance, and, secondly, the influence of external disturbances on the platoon formation and stability. The current paper focusses on these aspects by the development of control strategies that rely on the introduction of a novel delay-based spacing policy and a new definition of platoon stability (which will be referred to as disturbance string stability) that explicitly takes external disturbances into account.

The constant spacing policy and constant headway policy are most commonly considered in the literature, where the former requires a constant distance between two successive vehicles \cite{swaroop_1999}. The constant headway policy \cite{ioannou_1993} relaxes this requirement by including a dependence on the velocity of the follower vehicle. A comparison can be found in \cite{swaroop_1994}, whereas nonlinear spacing policies are given in \cite{yanakiev_1998}. However, these spacing policies are typically employed under the implicit assumption that the platoon should track a constant reference velocity. The tracking of varying reference velocity profiles has received considerably less attention, even though this is desirable in many practical situations. An important example is given by a heavy-duty vehicle traversing a road segment with varying road topography, for which it is known that the fuel-optimal velocity profile is generally varying \cite{hellstrom_2009}. For a platoon traversing a hilly road segment, it is desirable for each vehicle to track the same velocity profile in the spatial domain (i.e., relative to the position on the road). This is however incompatible with the constant spacing and constant headway strategies, for which the specification on the inter-vehicular distance might require vehicles to accelerate while climbing a hill. As this is potentially infeasible due to limited engine power, this leads to unsatisfactory platoon behavior, as has been recognized in experiments \cite{alam_2015b}. A delay-based spacing policy is introduced in this paper that guarantees that all vehicles track the same velocity profile in the spatial domain.

Stability analysis of interconnected systems such as vehicle platoons generally relies on notions of string stability, which characterizes the amplification (or, in fact, the lack thereof) of disturbances through the group (string) of vehicles. A formal definition is given in \cite{swaroop_1996} by requiring uniform boundedness of the states of all systems (see \cite{pant_2002} for a generalization towards higher spatial dimensions). Alternative definitions require bounds on the amplification of perturbations as a measure of string stability, e.g., \cite{fenton_1968,peppard_1974,sheikholeslam_1993,eyre_1998,liang_1999}, but these notions are typically only defined for linear systems. For an overview of string stability properties, see \cite{ploeg_2014}. Note that these references either consider autonomous systems or interconnected systems in which only the lead vehicle in a platoon is subject to external disturbances. The practically relevant case in which each vehicle is subject to external disturbances is considered in~\cite{seiler_2004}, whereas extensions are presented in~\cite{middleton_2010} and~\cite{peters_2014}. However, the analysis presented in these works relies on a transfer function approach and is therefore only applicable to linear systems. Moreover, in these works, only input disturbances are considered and the effect of initial conditions is not taken into account.

In the current paper, a definition of disturbance string stability for interconnected systems is introduced that explicitly includes the effects of initial condition perturbations and external disturbances on each vehicle. This notion provides a direct extension of the definition in \cite{swaroop_1996} by using the input-to-state stability framework introduced in \cite{sontag_1989}. Specifically, disturbance string stability can be regarded as a uniform (over the vehicle index) input-to-state stability property and applies also to nonlinear systems. It extends the notion of leader-to-formation stability in \cite{tanner_2004b} to platoons with external disturbances that are not limited to the leader.

The main contributions of this paper are as follows. First, a novel delay-based spacing policy is presented that guarantees that all vehicles in a platoon track a desired (and spatially varying) reference velocity profile. Second, the notion of disturbance string stability is introduced as a relevant stability property for interconnected systems subject to external disturbances. Third, on the basis of these definitions, a controller design method is presented that guarantees the tracking of the delay-based spacing policy and guarantees disturbance string stability with respect to the varying reference velocity. This design is performed in the spatial domain rather than the time domain, which leads to a simple design procedure that avoids the use of delay-dependent synthesis techniques. Using this controller design it is shown that string stability follows from a suitable choice of the spacing policy rather than the exact choice of the controller, which is the fourth contribution of this paper. Preliminary results on platoon control using a delay-based spacing policy can be found in \cite{besselink_2015}.

The remainder of this paper is outlined as follows. In Section~\ref{sec_spacingpolicies}, existing spacing policies are discussed and a motivation is provided for the introduction of the delay-based spacing policy used in this paper. Next, Section~\ref{sec_stringstab} introduces the notion of disturbance string stability and provides results that guarantee disturbance string stability of platoons on the basis of local properties associated to single vehicles. A controller that tracks the desired spacing policy is discussed in Section~\ref{sec_platooncontrol} and its disturbance string stability properties are shown. The results are illustrated by means of an example in Section~\ref{sec_evaluation} before conclusions are stated in Section~\ref{sec_conclusions}.

\textit{Notation.} The field of real numbers is denoted by $\R$, whereas $\N = \{1,2,\ldots\}$. For a vector $x\in\R^n$, its Euclidian norm is given as $|x| = \sqrt{x^{\T}x}$. Given a signal $x:\Tcal\rightarrow\R^n$, $\infnorm{x}{\Tcal}$ denotes its $\Linf$ norm defined as $\infnorm{x}{\Tcal} = \sup_{t\in\Tcal} |x(t)|$, where the shorthand notation $\|x\|_{\infty} = \infnorm{x}{[0,\infty)}$ is used when $\Tcal = [0,\infty)$. A continuous function $\alpha:[0,a)\rightarrow[0,\infty)$ is said to be of class $\classK$ if it is strictly increasing and $\alpha(0) = 0$. If, in addition, $a=\infty$ and $\alpha(r)\rightarrow\infty$ as $r\rightarrow\infty$, it is of class $\classKinf$. A continuous function $\beta:[0,a)\times[0,\infty)\rightarrow[0,\infty)$ is said to be of class $\classKL$ if, for each fixed $s$, the function $\beta(\cdot,s)$ is of class $\classK$ and, for each fixed $r$, $\beta(r,\cdot)$ is decreasing and satisfies $\beta(r,s)\rightarrow0$ as $s\rightarrow\infty$.

\section{Spacing policies and motivation}\label{sec_spacingpolicies}
The definition of the spacing policy has a crucial impact on the dynamic behavior of platoons of closely-spaced vehicle. Figure~\ref{fig_spacing_trucks} illustrates the spacing between two vehicles $i-1$ and $i$ in a platoon. In the literature, several different spacing policies have been proposed, of which the \emph{constant spacing} policy and the \emph{constant headway} policy are the most notable. These policies are shortly reviewed in this section, providing a motivation for a novel spacing policy as analyzed in the remainder of this paper: the \emph{delay-based} spacing policy.

\begin{figure}
\begin{center}
\includegraphics[scale=1]{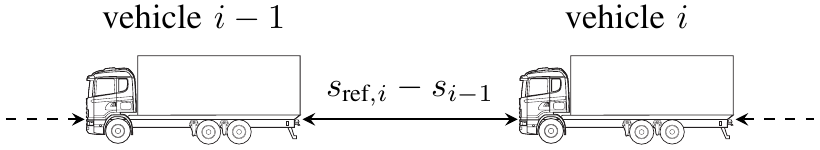}
\caption{Desired spacing policy $s_{\sref,i}(t) - s_{i-1}(t)$ between automatically controlled vehicles in a platoon.}
\label{fig_spacing_trucks}
\end{center}
\end{figure}

Let $s_i$ denote the longitudinal position of vehicle $i$ and $v_i$ its velocity. Naturally, they satisfy the kinematic relation
\begin{align}
\dot{s}_i(t) \defl \frac{\dd s_i}{\dd t}(t) = v_i(t).\label{eqn_kinematics}
\end{align}
A spacing policy describes the desired behavior $s_{\sref,i}(t)$ of vehicle $i$ on the basis of its predecessor $i-1$. Figure~\ref{fig_spacingpolicies} depicts the velocity of all vehicles in a platoon for the constant spacing, constant headway, and delay-based spacing policies. Here, it is assumed that the velocity of the lead vehicle $v_0(t)$ is prescribed and all follower vehicles track the desired behavior perfectly, i.e., $s_i(t) = s_{\sref,i}(t)$. The policies are further described next.

\begin{figure}
  \begin{center}
	\includegraphics[scale=1]{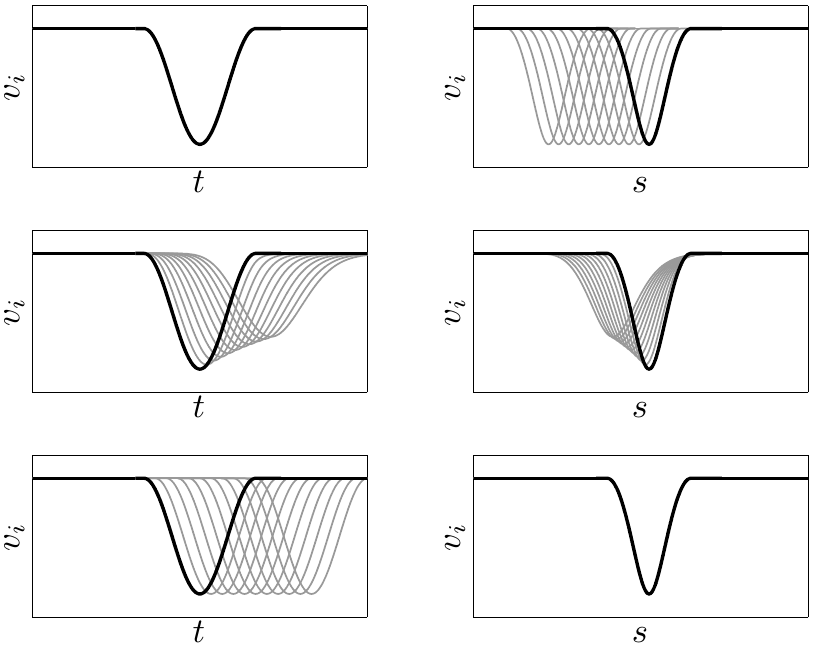}
    \caption{Velocities $v_i$ of ten follower vehicles (gray) as a result of a predefined velocity profile of the lead vehicle (black) for a constant spacing policy (top row), constant headway policy (middle row), and delay-based policy (bottom row). The left column shows the velocity as a function of time $t$, whereas the right column gives the velocity as a function of space $s$.}
    \label{fig_spacingpolicies}
  \end{center}
\end{figure}

The constant spacing policy (e.g., \cite{swaroop_1999}) takes the form
\begin{align}
s_{\sref,i}(t) = s_{i-1}(t) - d, \label{eqn_constantspacing}
\end{align}
where $d\geq0$ is the desired inter-vehicular distance. By using the assumption $s_i(t) = s_{\sref,i}(t)$ and (\ref{eqn_kinematics}), the policy (\ref{eqn_constantspacing}) implies that changes in velocity occur simultaneously in time (i.e., $v_i(t) = v_{i-1}(t)$). This is also apparent from the top left graph in Figure~\ref{fig_spacingpolicies}. If the change in velocity of the lead vehicle is the result of a disturbance, it is clear that the effect of this disturbance is not suppressed throughout the string. In fact, it has been shown in \cite{seiler_2004} that disturbance attenuation cannot be obtained for any linear controller that only uses measurements of the preceding vehicle $i-1$ for the control of vehicle~$i$.

An alternative spacing policy that inherently attenuates the effect of disturbances is given by the constant headway policy (e.g., \cite{swaroop_1994,ioannou_1993}), which includes a dependence on the velocity $v_i$ as
\begin{align}
s_{\sref,i}(t) = s_{i-1}(t) - (d + hv_i(t)), \label{eqn_constantheadway}
\end{align}
with $h>0$. By again using $s_i(t) = s_{\sref,i}(t)$ and (\ref{eqn_kinematics}), this can be written as
\begin{align}
h\dot{s}_i(t) = -s_i(t) + s_{i-1}(t) - d,
\end{align}
which shows that the desired reference position is essentially obtained by application of a first-order filter to the position of the preceding vehicle. It is this filtering, which is also apparent from the middle left graph in Figure~\ref{fig_spacingpolicies}, that is responsible for the inherent attenuation of disturbances for platoon controllers based on the constant headway policy.

However, it is clear from the graphs in the right column of Figure~\ref{fig_spacingpolicies} that, for the constant spacing and constant headway spacing policies, the changes in velocity occur on different positions in space for successive vehicles in the platoon. If the velocity change of the first vehicle was due to road properties such as hills rather than small undesired disturbances, this is potentially a large disadvantage. To illustrate this, consider a platoon of heavy-duty vehicles climbing a hill. Due to limited engine power, a large gradient can cause the lead vehicle of the platoon to decrease velocity as in Figure~\ref{fig_spacingpolicies}. In this case, follower vehicles might be required to have a higher velocity on this hill (i.e., at the same location in space) when they are subject to a constant spacing or constant headway policy. This might be infeasible due to limited engine power and leads to undesired platoon behavior and increased fuel consumption, as recognized in \cite{alam_2013} and \cite{turri_2016}.

In this paper, a spacing policy is introduced that guarantees that vehicles track the same velocity profile \emph{in space}, which avoids the aforementioned disadvantages. The delay-based spacing policy is given as
\begin{align}
s_{\sref,i}(t) = s_{i-1}(t - \dt), \label{eqn_constanttimegap}
\end{align}
where vehicle $i$ tracks a time-delayed version of the trajectory of the preceding vehicle, with time gap $\dt>0$ (see also \cite{newell_2002}). The policy (\ref{eqn_constanttimegap}) indeed achieves equal velocity profiles in space, as formalized in the following proposition.
\begin{proposition}\label{lem_timegap}
Consider the kinematics (\ref{eqn_kinematics}) and assume $s_i(t) = s_{\sref,i}(t)$ and $v_i(t)>0$ for all $t\in\R$. Then, (\ref{eqn_constanttimegap}) holds if and only if\footnote{The slight abuse of notation $v_i(t)$ and $v_i(s)$ will be used to indicate the velocity of vehicle $i$ as a function of time and space, respectively.}, for some function $\vref(\cdot)$,
\begin{align}
v_i(s) = v_{i-1}(s) = \vref(s).
\end{align}
\end{proposition}
\begin{proof}
In order to prove the proposition, let $s$ be a point in space and let $t_i(s)$ be the time instance when vehicle $i$ passes that point. Note that the assumption $v_i(t)>0$ for all $t\in\R$ guarantees that $t_i(s)$ is uniquely defined. Then, using $s_i(t) = s_{\sref,i}(t)$, (\ref{eqn_constanttimegap}) can equivalently be written as
\begin{align}
t_i(s) = t_{i-1}(s) + \dt,
\label{eqn_prp_timegap_proof_step1}
\end{align}
for all $s\in\R$. Next, the expression of the kinematic relation (\ref{eqn_kinematics}) in spatial domain leads to
\begin{align}
\frac{\dd t_i}{\dd s}(s) = \frac{1}{v_i(s)},
\end{align}
after which integration yields
\begin{align}
t_i(s_1) - t_i(s_0) = \int_{s_0}^{s_1} \frac{1}{v_i(s)}\di s,
\label{eqn_prp_timegap_proof_step2}
\end{align}
for some initial position $s_0$. When considering (\ref{eqn_prp_timegap_proof_step2}) for vehicles $i$ and $i-1$, the subtraction of both results and use of (\ref{eqn_prp_timegap_proof_step1}) leads to
\begin{align}
\int_{s_0}^{s_1} \frac{1}{v_i(s)} - \frac{1}{v_{i-1}(s)} \di s
=  \dt - \dt = 0.
\label{eqn_prp_timegap_proof_step3}
\end{align}
As (\ref{eqn_prp_timegap_proof_step3}) holds for all $s_0,s_1\in\R$ such that $s_1\geq s_0$, it is clear that $v_i(s) = v_{i-1}(s) \defr \vref(s)$ for all $s$, proving the first part of the proposition.

To prove the converse, assume that $v_i(s) = v_{i-1}(s) = \vref(s)$. Subsitution of this in the left-hand term in (\ref{eqn_prp_timegap_proof_step3}) gives $t_i(s_1) - t_{i-1}(s_1) = t_i(s_0) - t_{i-1}(s_0) \defr \dt$, finalizing the proof.
\end{proof}

Motivated by the discussion above, the objective of this paper is the design of a controller that, first, achieves asymptotic tracking of a spatially varying common reference velocity $\vref(\cdot)$ and delay-based spacing policy (\ref{eqn_constanttimegap}), and, second, guarantees string stability with respect to this desired trajectory and in the presence of external disturances. In order to achieve the latter, the notion of disturbance string stability is introduced in the next section.

\section{String stability analysis with disturbances}\label{sec_stringstab}
Consider a platoon of automatically controlled vehicles as represented through the autonomous cascaded interconnection
\begin{align}
\begin{array}{rcl}
\dot{x}_0 &=& f(x_0,0), \\
\dot{x}_i &=& f(x_i,x_{i-1}), \quad i\in\Ical_N,
\end{array}\label{eqn_cascsys}
\end{align}
where $\Ical_N = \{1,2,\ldots,N\}$. In the context of platooning, $\Ical_N$ represents the set of follower vehicles, whereas $\Ical_N^0 = \{0,1,\ldots,N\}$ includes the lead vehicle with index $0$. In (\ref{eqn_cascsys}), $x_i\in\R^n$, $i\in\Ical_N^0$ is the state of the system and the function $f:\R^n\times\R^n\rightarrow\R^n$ satisfying $f(0,0) = 0$ is assumed to be locally Lipschitz continuous in both arguments.

For such systems, the notion of string stability was introduced in \cite{swaroop_1996} according to the following definition:
\begin{definition}\label{def_stringstab}
The equilibrium $x_i = 0$, $i\in\Ical_N^0$, of the platoon (\ref{eqn_cascsys}) is said to be string stable if, for any $\varepsilon > 0$, there exists a $\delta > 0$ such that, for all $N\in\N$,
\begin{align}
\sup_{i\in\Ical_N^0} |x_i(0)| < \delta \;\Rightarrow\; \sup_{i\in\Ical_N^0} |x_i(t)| < \varepsilon, \;\; \forall t\geq0.
\end{align}
\end{definition}
Asymptotic string stability is defined in \cite{swaroop_1996} as follows:
\begin{definition}\label{def_stringstab_asymp}
The equilibrium $x_i = 0$, $i\in\Ical_N^0$, of the platoon (\ref{eqn_cascsys}) is said to be asymptotically string stable if it is string stable and $\delta$ can be chosen such that
\begin{align}
\sup_{i\in\Ical_N^0} |x_i(0)| < \delta \;\Rightarrow\; \lim_{t\rightarrow\infty}\sup_{i\in\Ical_N^0} |x_i(t)| = 0.
\end{align}
\end{definition}
\begin{remark}\label{rem_stringstab_classKLfunc}
Definitions~\ref{def_stringstab} and~\ref{def_stringstab_asymp} are similar to the standard notion of Lyapunov stability, with the difference that, in the former, perturbations from the equilibrium are measured as the worst-case perturbation over all subsystems. Nonetheless, by exploiting this similarity, it follows directly that asymptotic string stability can equivalently be expressed through the existence of a function $\bar{\beta}$ of class $\classKL$ and a constant $\bar{c}>0$ such that, for all $N\in\N$,
\begin{align}
\sup_{i\in\Ical_N^0} |x_i(t)| \leq \bar{\beta}\!\left(\sup_{i\in\Ical_N^0} |x_i(0)|,t\right), \; \forall \sup_{i\in\Ical_N^0} |x_i(0)| < \bar{c},
\end{align}
and for all $t\geq0, $see, e.g., \cite{book_khalil_2002}.
\end{remark}

The notions of string stability in Definitions~\ref{def_stringstab} and~\ref{def_stringstab_asymp} apply to autonomous interconnected systems of the form (\ref{eqn_cascsys}). However, in many practical situations, vehicles are subject to external disturbances. Therefore, the following non-autonomous platoon
\begin{align}
\begin{split}
\frac{\dd x_0}{\dd\theta} &= f(x_0,0,w_0), \\
\frac{\dd x_i}{\dd\theta} &= f(x_i,x_{i-1},w_i), \quad i\in\Ical_N
\end{split}\label{eqn_cascsys_dist}
\end{align}
is considered, where $w_i\in\R^m$, $i\in\Ical_N^0$, represent disturbances influencing the system. Moreover, $\theta$ is taken as the independent variable in (\ref{eqn_cascsys_dist}), which is motivated by the observation that vehicle dynamics can be expressed in either time domain or the spatial domain. Consequently, $\theta$ can either represent time $t$ or the spatial variable $s$. The latter case will be further explored in controller design in Section~\ref{sec_platooncontrol}, as it provides a convenient approach for the synthesis of controllers that track the delay-based spacing policy (\ref{eqn_constanttimegap}).

The following definition of disturbance string stability is introduced to address the effects of disturbances in interconnected systems of the form (\ref{eqn_cascsys_dist}).
\begin{definition}\label{def_diststringstab}
The platoon (\ref{eqn_cascsys_dist}) is said to be disturbance string stable if there exist functions $\bar{\beta}$ of class $\classKL$ and $\bar{\sigma}$ of class $\classKinf$ and constants $\bar{c}>0$, $\bar{c}_w>0$, such that, for any initial condition $x_i(\theta_0)$ and disturbance $w_i$, $i\in\Ical_N^0$, satisfying
\begin{align}
\sup_{i\in\Ical_N^0} |x_i(\theta_0)| < \bar{c}, \quad \sup_{i\in\Ical_N^0} \|w_i\|_{\infty} < \bar{c}_w,
\label{eqn_def_diststringstab_bounds}
\end{align}
the solution $x_i(\theta)$, $i\in\smash{\Ical_N^0}$, exists for all $\theta\geq\theta_0$ and satisfies
\begin{align}
\sup_{i\in\Ical_N^0} |x_i(\theta)|
&\leq \bar{\beta}\!\left(\sup_{i\in\Ical_N^0} |x_i(\theta_0)|,\theta-\theta_0\right) \nonumber\\
&\qquad+ \bar{\sigma}\!\left(\sup_{i\in\Ical_N^0} \|w_i\|_{\infty}^{[\theta_0,\theta]}\right), \;\; \forall N \in\N.
\label{eqn_def_diststringstab}
\end{align}
If $\bar{c}$ and $\bar{c}_w$ can be taken as $\bar{c} = \infty$, $\bar{c}_w=\infty$, then the platoon~(\ref{eqn_cascsys_dist}) is said to be globally disturbance string stable.
\end{definition}
In the absence of disturbances, the definition of disturbance string stability in Definition~\ref{def_diststringstab} is equivalent to the notion of asymptotic string stability in Definition~\ref{def_stringstab_asymp}. Moreover, it extends the definition of string stability in \cite{ploeg_2014} by allowing for disturbances on all vehicles rather than the lead vehicle only and explicitly captures the effects of initial conditions.

It is noted that condition (\ref{eqn_def_diststringstab}) in Definition~\ref{def_diststringstab} is required to hold for any string length $N\in\N$, rather than for fixed $N$ corresponding to the length of the platoon under consideration. The invariance of the bounds under the string length is an important property, as it guarantees that the notion of disturbance string stability is scalable and allows for the addition or removal of vehicles from a string without affecting stability (see also \cite{ploeg_2014,yadlapalli_2006}). In fact, it states that the state trajectories remain bounded for any $N\in\N$, which prohibits the amplification of disturbances as they propagate through the platoon.

The definition of disturbance string stability in Definition~\ref{def_diststringstab} is based on properties of the entire platoon. The following theorem allows for establishing disturbance string stability of the basis of \emph{local} properties and is the main result of this section.
\begin{theorem}\label{thm_localiss}
Consider the platoon (\ref{eqn_cascsys_dist}) and let each vehicle be input-to-state stable with respect to its inputs $x_{i-1}$ and $w_i$, i.e., there exist a function $\beta$ of class $\classKL$, functions $\gamma$ and $\sigma$ of class $\classKinf$ and constants $c>0$, $c_w>0$, such that trajectories $x_i$ satisfy
\begin{align}
|x_i(\theta)|
&\leq \beta\big(|x_i(\theta_0)|,\theta-\theta_0\big) + \gamma\!\left(\|x_{i-1}\|_{\infty}^{[\theta_0,\theta]}\right) \nonumber\\
&\qquad+ \sigma\!\left(\|w_i\|_{\infty}^{[\theta_0,\theta]}\right), \;\;\forall \theta\geq\theta_0,
\label{eqn_thm_localiss}
\end{align}
for any $|x_i(\theta_0)| < c$, $\|x_{i-1}\|_{\infty} < c$, $\|w_i\|_{\infty} < c_w$ and for all $i\in\Ical^0_N$ and $N\in\N$ (with $x_{i-1} = 0$ for $i=0$). If the function $\gamma$ satisfies $\gamma(r)\leq\bar{\gamma}r$ for all $r\geq0$ and for some $\bar{\gamma} < 1$, then the platoon (\ref{eqn_cascsys_dist}) is disturbance string stable. If, in addition, the function $\beta$ in (\ref{eqn_thm_localiss}) satisfies
\begin{align}
\beta(r,\omega s) \leq \frac{1}{\omega^q}\beta(r,s), \quad \forall\, r,s\geq0,
\label{eqn_thm_localiss_beta}
\end{align}
for all $\omega$, $0<\omega\leq1$ and some $q>0$, then the function $\bar{\beta}$ in (\ref{eqn_def_diststringstab}) can be taken of the form $\bar{\beta}(r,s) = c_{\beta}\beta(\alpha(r),s)$ for some constant $c_{\beta}>0$ and function $\alpha$ of class $\classKinf$. Finally, if $c$ and $c_w$ can be chosen as $c = \infty$, $c_w=\infty$, then the platoon (\ref{eqn_cascsys_dist}) is globally disturbance string stable.
\end{theorem}
\begin{proof}
The proof is given in Appendix~\ref{app_thm_localiss_proof}.
\end{proof}
\begin{remark}
The result in Theorem~\ref{thm_localiss} ensures that the influence of the initial condition does not vanish arbitrarily slow. Roughly speaking, (\ref{eqn_thm_localiss_beta}) characterizes functions $\beta$ that have a convergence rate slower than $\theta^{-q}$ (for some $q>0$) and it is shown that for such $\beta$ the function $\bar{\beta}$ in (\ref{eqn_def_diststringstab}) has the same convergence rate. Even though an upper bound satisfying (\ref{eqn_thm_localiss_beta}) is used when $\beta$ has a faster convergence rate (see (\ref{eqn_thm_localiss_proof_phi}) in the proof in Appendix~\ref{app_thm_localiss_proof}), this ensures that that the function $\bar{\beta}$ in (\ref{eqn_def_diststringstab}) does not have arbitrarily slow convergence, which is not a priori obvious when the number of interconnected systems increases.
\end{remark}
\begin{remark}
At first sight it might be surprising that a proof of Theorem~\ref{thm_localiss} is required, as it is well-known that the cascade interconnection of input-to-state stable systems is itself input-to-state stable. However, this standard result in, e.g., \cite{sontag_1989,book_krstic_1995}, does not guarantee a priori that the class $\classKL$ and class $\classKinf$ functions that bound the behavior of the cascaded system remain bounded when the number of interconnected systems grows. For example, a cascade of (linear) systems $\dot{x}_{i} = -x_i + 2x_{i-1} + w_i$ is clearly input-to-state stable, but it can be shown that perturbations can grow unbounded as the number of subsystems $N$ grows (consider, e.g., the static behavior for $w_i = 1$ for all $i$). Theorem~\ref{thm_localiss} explicitly addresses this aspect.
\end{remark}

The result in Theorem~\ref{thm_localiss} deals with subsystems that are connected through their entire states $x_i$. However, a practically relevant case is given by systems of the form
\begin{align}
\begin{split}
\frac{\dd x_0}{\dd\theta} &= f(x_0,0,w_0), \\
\frac{\dd x_i}{\dd\theta} &= f(x_i,y_{i-1},w_i), \quad i\in\Ical_N,
\end{split}\label{eqn_cascsys_dist_output}
\end{align}
in which the interconnection is achieved through outputs $y_i\in\R^p$ defined as
\begin{align}
y_i &= h(x_i), \quad i\in\Ical_N^0.
\label{eqn_cascsys_dist_output_h}
\end{align}
The interconnection of the form (\ref{eqn_cascsys_dist_output})--(\ref{eqn_cascsys_dist_output_h}) can be studied by exploiting the notion of input-to-output stability \cite{sontag_2008}. This leads to the following theorem, which can be regarded as a counterpart of Theorem~\ref{thm_localiss} for systems with interconnection through the outputs.
\begin{theorem}\label{thm_localios}
Consider the platoon (\ref{eqn_cascsys_dist_output})--(\ref{eqn_cascsys_dist_output_h}) and let each vehicle be input-to-output stable with respect to its inputs $y_{i-1}$ and $w_i$, i.e., there exist a function $\beta_y$ of class $\classKL$, functions $\gamma_y$ and $\sigma_y$ of class $\classKinf$ and constants $c>0$, $c_w>0$, such that the outputs $y_i = h(x_i)$ satisfy
\begin{align}
|y_i(\theta)| &\leq \beta_y\big(|x_i(\theta_0)|,\theta-\theta_0\big)
+ \gamma_y\!\left(\|y_{i-1}\|_{\infty}^{[\theta_0,\theta]}\right) \nonumber\\
&\qquad + \sigma_y\!\left(\|w_i\|_{\infty}^{[\theta_0,\theta]}\right), \;\forall \theta\geq\theta_0,
\label{eqn_thm_localios}
\end{align}
for any $|x_i(\theta_0)| < c$, any $y_{i-1} = h(x_{i-1})$ with $\|x_{i-1}\|_{\infty}<c$, $\|w_i\|_{\infty}<c_w$ and for all $i\in\Ical_N^0$ and $N\in\N$ (with $y_{i-1} = 0$ for $i=0$).
Moreover, let each vehicle in (\ref{eqn_cascsys_dist_output}) be input-to-state stable with respect to the same inputs, i.e., there exist a function $\beta_x$ of class $\classKL$ and functions $\gamma_x$ and $\sigma_x$ of class $\classKinf$ such that
\begin{align}
|x_i(\theta)| &\leq \beta_x\big(|x_i(\theta_0)|,\theta-\theta_0\big)
+ \gamma_x\!\left(\|y_{i-1}\|_{\infty}^{[\theta_0,\theta]}\right) \nonumber\\
&\qquad+ \sigma_x\!\left(\|w_i\|_{\infty}^{[\theta_0,\theta]}\right), \;\forall \theta\geq\theta_0,
\label{eqn_thm_localios_iss}
\end{align}
for any $|x_i(\theta_0)| < c$, any $y_{i-1} = h(x_{i-1})$ with $\|x_{i-1}\|_{\infty}<c$, $\|w_i\|_{\infty}<c_w$ and for all $i\in\Ical_N^0$ and $N\in\N$.
If the function $\gamma_y$ satisfies $\gamma_y(r)\leq\bar{\gamma}r$ for all $r\geq0$ and for some $\bar{\gamma} < 1$, then the platoon (\ref{eqn_cascsys_dist_output})--(\ref{eqn_cascsys_dist_output_h}) is disturbance string stable.
\end{theorem}
\begin{proof}
The proof can be found in Appendix~\ref{app_thm_localios_proof}.
\end{proof}
At first sight, the conditions in Theorem~\ref{thm_localios} seem more restrictive than those in Theorem~\ref{thm_localiss}, as input-to-state stability of the subsystems is required in both cases. However, the gain function $\gamma_y$ of the input-to-output stability property in (\ref{eqn_thm_localios}) is required to be bounded as $\gamma_y(r)\leq\bar{\gamma}r$ for some $\bar{\gamma}<1$, whereas the gain function $\gamma_x$ in (\ref{eqn_thm_localios_iss}) can be arbitrarily large. Thus, Theorem~\ref{thm_localios} shows that only the gain with respect to the interconnection variables is relevant in proving disturbance string stability. In addition, this input-to-output gain $\gamma_y$ is typically smaller than the input-to-state gain $\gamma_x$, providing less conservative results.
\begin{remark}
The condition (\ref{eqn_thm_localios_iss}) is required to provide a bound on the state trajectories whenever the interconnection variables $y_i$ remain bounded, which is of importance as disturbance string stability is defined on the basis of state trajectories. Here, it is remarked that input-to-state stability as in (\ref{eqn_thm_localios_iss}) can be implied by input-to-output stability as in (\ref{eqn_thm_localios}) when the subsystems satisfy observability properties that are relevant in the input-to-state stability framework. The notion of input/output-to-state stability in \cite{sontag_1997} (see also \cite{sontag_2008}) provides such a property.
\end{remark}

\section{Platoon control for disturbance string stability}\label{sec_platooncontrol}
Vehicle platoons are considered in this section and a class of controllers is synthesized that achieves tracking of the delay-based spacing policy (\ref{eqn_constanttimegap}) and guarantees disturbance string stability. Thereto, vehicle modeling is discussed in Section~\ref{sec_modeling}, before presenting controller design in Section~\ref{sec_controldesign}. The resulting closed-loop stability properties are analyzed in Section~\ref{sec_platoonstability}.

\subsection{Platoon modeling and objectives}\label{sec_modeling}
Consider a platoon of $N+1$ vehicles, in which each vehicle satisfies the longitudinal dynamics
\begin{align}
\begin{split}
\dot{s}_i(t) &= \tilde{h}(\xi_i(t)), \\
\dot{\xi}_i(t) &= \tilde{f}(\xi_i(t)) + \tilde{g}(\xi_i(t))u_i(t) + \tilde{p}(\xi_i(t))w_i(t),
\end{split}\label{eqn_platoonsys_time}
\end{align}
with $i\in\Ical_N^0$. Here, $s_i(t)\in\R$ denotes the position of vehicle $i$, such that the first equation in (\ref{eqn_platoonsys_time}) represents the kinematic relation with velocity $v_i(t)\defl\smash{\tilde{h}(\xi_i(t))}$. The second equation with state $\xi(t)\in\R^{n-1}$ is a general description of the remaining dynamics, which can include engine or drive train dynamics as well as low-level control systems. The input $u_i\in\R$ is available for the platoon control developed in this section, whereas $w_i\in\R^m$ is the unmeasurable external disturbance. It is assumed that the functions $\tilde{f}:\R^{n-1}\rightarrow\R^{n-1}$, $\tilde{g}:\R^{n-1}\rightarrow\R^{n-1}$, $\tilde{p}:\R^{n-1}\rightarrow\R^{(n-1)\times m}$ and $\tilde{h}:\R^{n-1}\rightarrow\R$ are sufficiently smooth.

The dynamics (\ref{eqn_platoonsys_time}) is taken to satisfy the following assumption, which simplifies the developments (see, e.g., \cite{book_khalil_2002} for a definition of relative degree).
\begin{assumption}\label{ass_relativedegree}
The dynamics (\ref{eqn_platoonsys_time}) with input $u_i$ has relative degree $n$ with respect to the output~$s_i$.
\end{assumption}
\begin{remark}
Vehicle models commonly considered in the analysis and control of vehicle platoons typically satisfy Assumption~\ref{ass_relativedegree}. Specifically, the second-order models used in \cite{levine_1966,melzer_1971} as well as third-order models (e.g., including actuator dynamics) considered in \cite{ioannou_1993,swaroop_1994,ploeg_2014b} are of the form~(\ref{eqn_platoonsys_time}) and satisfy this assumption.
\end{remark}
\begin{remark}
The disturbances $w_i$ in (\ref{eqn_platoonsys_time}) are taken as external disturbances, but they might as well result from modeling errors or parameter uncertainties. Moreover, even though it is assumed that all vehicles have identical dynamics (\ref{eqn_platoonsys_time}), the results in this paper have the potential to be extended to heterogeneous vehicle platoons. Namely, the disturbance~$w_i$ might be the result of model inhomogeneity rather than external influences.
\end{remark}

Motivated by the discussion in Section~\ref{sec_spacingpolicies}, a controller will be synthesized that, firstly, achieves the desired inter-vehicular spacing according to the delay-based policy (\ref{eqn_constanttimegap}), secondly, ensures tracking of a common velocity profile $\vref(\cdot)$ in space, and, thirdly, guarantees disturbance string stability with respect to this velocity profile. Here, it is recalled that the first two objectives are aligned for positive velocities according to Proposition~\ref{lem_timegap}. Therefore, the following assumption is made on the reference velocity.
\begin{assumption}\label{ass_vref}
The reference velocity $\vref(\cdot)$ satisfies $0< v_{\smin} \leq \vref(s)\leq v_{\smax}$ for all $s\geq0$ and for some constants $v_{\smin}$, $v_{\smax}$. Moreover, $\vref(\cdot)$ is at least $n-2$ times continuously differentiable.
\end{assumption}
\begin{remark}
In addition to allowing for expressing the spacing policy (\ref{eqn_constanttimegap}) in the spatial domain as will be used in the remainder of this paper, the assumption $v_{\smin} \leq \vref(s)$ for positive $v_{\smin}$ guarantees that the follower distance $d_i(t) = s_{i-1}(t) - s_i(t)$ remains positive as long as the reference velocity is perfectly tracked. In fact, the follower distance satisfies $v_{\smin}\Delta t \leq d_i(t) \leq v_{\smax}\Delta t$, providing a bound on the follower distance during maneuvers. Similarly, if the minimum velocity and nominal time gap $\Delta t$ are chosen such that $L_{\smax} \leq v_{\smin}\Delta t$, with $L_{\smax}$ the maximum vehicle length, subsequent vehicles do not collide when they perfectly track the reference velocity. It is recalled that the main benefit of platooning for (heavy-duty) vehicles, i.e., reduced aerodynamic drag, is only obtained for significantly large vehicle speeds, such that Assumption~\ref{ass_vref} is not too restrictive. In practice, the reference velocity profile $\vref(\cdot)$ should be designed such that the actuation constraints (i.e., bounds on traction force and braking capacity) of the vehicles are satisfied. Note that the fact that all vehicles track the same velocity profile in the spatial domain enables such design.
\end{remark}

Under the assumption that all vehicles have a positive velocity at all times, the delay-based spacing policy (\ref{eqn_constanttimegap}) can equivalently be expressed in the spatial domain. Thereto, let the space $s$ be the independent variable and denote $t_i(s)$ as the time instance at which vehicle $i$ passes $s$. Then, the spacing policy (\ref{eqn_constanttimegap}) can be represented as $\Delta_i(s) = 0$, where $\Delta_i$ denotes the deviation from the nominal time gap $\dt$ as
\begin{align}
\Delta_i(s) &= t_i(s) - t_{i-1}(s) - \dt, \label{eqn_Delta}\\
\Delta_i^0(s) &= t_i(s) - t_0(s) - i\dt, \label{eqn_Delta0}
\end{align}
for all $i\in\Ical_N$. Similarly, $\Delta_i^0$ represents the deviation from the nominal time gap with respect to the first vehicle in the platoon. As the characterization in (\ref{eqn_Delta}) does not require analysis of time-delay systems as suggested by (\ref{eqn_constanttimegap}), it is beneficial to consider controller synthesis in the spatial domain.

The vehicle dynamics (\ref{eqn_platoonsys_time}) can be written in spatial domain by exploiting the kinematic relation~(\ref{eqn_kinematics}), which leads to
\begin{align}
\begin{split}
\mathring{t}_i(s) &= h(\xi_i(s)), \\
\mathring{\xi}_i(s) &= f(\xi_i(s)) + g(\xi_i(s))u_i(s) + p(\xi_i(s))w_i(s),
\end{split}\label{eqn_platoonsys_space}
\end{align}
with $\mathring{x}(s)\defl\tfrac{\dd x}{\dd s}(s)$ denoting the derivative with respect to space and for all $i\in\Ical_N^0$. Moreover,
\begin{align}
h(\xi_i) &= \frac{1}{\tilde{h}(\xi_i)}, \;\;
f(\xi_i) = \frac{\tilde{f}(\xi_i)}{\tilde{h}(\xi_i)}, \nonumber\\
g(\xi_i) &= \frac{\tilde{g}(\xi_i)}{\tilde{h}(\xi_i)}, \;\;
p(\xi_i) = \frac{\tilde{p}(\xi_i)}{\tilde{h}(\xi_i)}.
\label{eqn_platoonsys_space_functions}
\end{align}
Contrary to the description in (\ref{eqn_platoonsys_time}), the disturbance $w_i$ is assumed to be specified in space in (\ref{eqn_platoonsys_space}). This does not pose any limitations as this disturbance will later be characterized by its norm $\|w_i\|_{\infty}$, which is independent from the choice of independent variable.

\subsection{Platoon controller design}\label{sec_controldesign}
The representation of the vehicle dynamics in the spatial domain (\ref{eqn_platoonsys_space}) will be exploited in the current section to design a class of controllers that achieve the desired objectives of tracking a (spatially-varying) reference velocity and the delay-based spacing policy while guaranteeing disturbance string stability. To enable controller design, the platoon of vehicles (\ref{eqn_platoonsys_space}) with spacing policy (\ref{eqn_Delta}) will be represented in time gap tracking error coordinates, which will be based on a representation of the vehicles in velocity tracking coordinates. Herein, an input-output linearization approach will be exploited.

In order to achieve tracking of the reference velocity $\vref(\cdot)$, the velocity tracking error $e_{1,i}$ as well as its space derivatives are defined, for any follower vehicle $i\in\Ical_N$, as
\begin{align}
e_{1,i}(s) &\defl h(\xi_i(s)) - \frac{1}{\vref(s)}, \label{eqn_veltrackerror_e1}\\
e_{k,i}(s) &\defl L_f^{k-1}h(\xi_i) - \frac{\dd^{k-1}}{\dd s^{k-1}}\!\left(\frac{1}{\vref(s)}\right), \label{eqn_veltrackerror_ek}
\end{align}
$k = 2,3,\ldots,n-1$, where it is recalled that $h(\xi_i(s)) = \tfrac{1}{v_i(s)}$ due to (\ref{eqn_platoonsys_space_functions}).
In (\ref{eqn_veltrackerror_ek}), the notation $L_fh(\xi)$ denotes the Lie derivative of $h$ along $f$ (albeit applied in spatial domain), see \cite{book_khalil_2002,book_nijmeijer_1990} for a definition.

By Assumption~\ref{ass_relativedegree}, there exists a controller
\begin{align}
u_i(s) &= \frac{1}{L_gL_f^{n-2}h(\xi_i)}\bigg( -L_f^{n-1}h(\xi_i) \nonumber\\
&\qquad\hspace{2cm}+ \frac{\dd^{n-1}}{\dd s^{n-1}}\!\left(\frac{1}{\vref(s)}\right) + \bar{u}_i(s) \bigg)
\label{eqn_inoutlin_ui}
\end{align}
that achieves input-output linearization of (\ref{eqn_platoonsys_space}) with respect to the output $h(\xi_i)$ and the virtual input $\bar{u}_i$, such that the dynamics of (\ref{eqn_platoonsys_space}) with (\ref{eqn_inoutlin_ui}) can be written as
\begin{align}
\begin{split}
\mathring{t}_i(s) &= e_{1,i}(s) + \frac{1}{\vref(s)}, \\
\mathring{e}_i(s) &= Ae_i(s) + B\bar{u}_i(s) + \rho(\xi_i)w_i(s),
\end{split}\label{eqn_platoonsys_space_te}
\end{align}
for any vehicle $i\in\Ical_N^0$. Here, the linear dynamics for $e_i = [\begin{array}{cccc} e_{1,i} & \cdots &  e_{n-1,i}\end{array}]^{\T}$ is characterized by the matrices
\begin{align}
A = \left[\begin{array}{cccc}
0 & 1 &   & 0 \\
  & \smash{\ddots} & \smash{\ddots} & \\
  && 0 & 1 \\
0 & & & 0
\end{array}\right], \quad
B = \left[\begin{array}{c} 0 \\ \smash{\vdots} \\ 0 \\ 1 \end{array}\right],
\label{eqn_platoonsys_space_matrices}
\end{align}
whereas the disturbance $w_i$ influences the dynamics (\ref{eqn_platoonsys_space_te}) through the function $\rho = [\begin{array}{ccc} \rho_1^{\T} & \cdots & \rho_{n-1}^{\T} \end{array}]^{\T}$ defined as
\begin{align}
\rho_k(\xi_i) = L_pL_f^{k-1}h(\xi_i), \quad k = 1,2,\ldots,n-1.
\label{eqn_rhok}
\end{align}
In (\ref{eqn_rhok}), the argument $\xi_i$ is maintained for ease of notation, but it is noted that $\xi_i$ can be related to the states $e_i$ and reference velocity $\vref$ through (\ref{eqn_veltrackerror_e1})--(\ref{eqn_veltrackerror_ek}).

Based on the velocity tracking error $e_{1,i}$ in (\ref{eqn_veltrackerror_e1}), a characterization of the required spacing policy for follower vehicle $i\in\Ical_N$ is introduced by defining the time gap tracking error
\begin{align}
\delta_{1,i}(s) \defl (1-\kappa_0)\Delta_i(s) + \kappa_0\Delta_i^0(s) + \kappa e_{1,i}(s),
\label{eqn_delta1}
\end{align}
with $0\leq \kappa_0 < 1$ and $\kappa>0$ and where $\Delta_i$ and $\Delta_i^0$ are defined in (\ref{eqn_Delta}) and (\ref{eqn_Delta0}), respectively. It can be observed that $\delta_{1,i}$ presents a weighted combination of the timing error with respect to the preceding vehicle and the first vehicle in the platoon. Moreover, the additional term $\kappa e_{1,i}$ allows for the relaxation of the spacing policy when the vehicle (with index $i$) does not perfectly track the desired velocity reference (see (\ref{eqn_veltrackerror_e1})) and will be shown to ensure damping of perturbations similar to the case of a constant headway strategy in (\ref{eqn_constantheadway}). Namely, the inclusion of this term induces the dynamics
\begin{align}
\kappa\mathring{\Delta}_i(s) = -\Delta_i(s) + \delta_{1,i}(s) - \kappa_0\Delta_{i-1}^0(s) - \kappa e_{1,i-1},
\label{eqn_Delta_dynamics}
\end{align}
as can be observed by noting that $\mathring{\Delta}_i = e_{1,i}-e_{1,i-1}$ and $\Delta_i^0 = \Delta_i + \Delta_{i-1}^0$ (see (\ref{eqn_Delta})--(\ref{eqn_Delta0})). For later reference, the terms dependent on the preceding vehicle (with index $i-1$) are collected, for any $i\in\Ical_N^0$, as
\begin{align}
y_i(s) \defl -\kappa_0\Delta_i^0(s) - \kappa e_{1,i}(s).
\label{eqn_output_yi}
\end{align}

Returning to the definition of $\delta_{1,i}$ in (\ref{eqn_delta1}), additional time gap tracking error coordinates $\delta_{k,i}$ are defined accordingly~as
\begin{align}
\delta_{k,i}(s) &= (1-\kappa_0)(e_{k-1,i} - e_{k-1,i-1}) \nonumber\\
&\qquad+ \kappa_0(e_{k-1,i} - e_{k-1,0}) + \kappa e_{k,i},
\label{eqn_deltak}
\end{align}
where $k = 2,3,\ldots,n-1$ and for $i\in\Ical_N$. Then, the platoon dynamics can be equivalently represented in the timing error coordinates $x_i = [\begin{array}{cc} \Delta_i & \delta_i^{\T} \end{array}]^{\T}$, where $\Delta_i$ represents the desired delay-based spacing policy as in (\ref{eqn_Delta}) and $\delta_i = [\begin{array}{cccc} \delta_{1,i} & \cdots & \delta_{n-1,i}\end{array}]^{\T}$ is given through (\ref{eqn_delta1}) and (\ref{eqn_deltak}). In particular, after introducing the new virtual input $\tilde{u}_i$ by substituting
\begin{align}
\bar{u}_i(s) &= -\kappa^{-1}(1-\kappa_0)(e_{n-1,i} - e_{n-1,i-1}) \nonumber\\
&\qquad- \kappa^{-1}\kappa_0(e_{n-1,i} - e_{n-1,0}) + \tilde{u}_i(s),
\label{eqn_input_ubar}
\end{align}
into (\ref{eqn_platoonsys_space_te}), it can be shown that the dynamics of the follower vehicles $i\in\Ical_N$ in timing error coordinates $x_i$ takes the form
\begin{align}
\begin{split}
\mathring{x}_i(s) &= F\big(x_i(s),\tilde{u}_i(s),y_{i-1}(s), \bar{\rho}(\xi_i,\xi_{i-1},\xi_0)\bar{w}_i(s)\big), \\
y_i(s) &= H(x_i(s)),
\end{split}\label{eqn_platoonsys_space_x}
\end{align}
with $y_i$ as in (\ref{eqn_output_yi}). By recalling the dynamics for $\Delta_i$ in (\ref{eqn_Delta_dynamics}) and by exploiting the definitions (\ref{eqn_delta1}), (\ref{eqn_deltak}) as well as the dynamics (\ref{eqn_platoonsys_space_te}), it follows that the vector field $F$ is given as
\begin{align}
F(x_i,\tilde{u}_i,y_{i-1},\omega_i) = \left[\begin{array}{c}
\kappa^{-1}(-\Delta_i + \delta_{1,i} + y_{i-1}) \\
A\delta_i + \kappa B\tilde{u}_i + \omega_i
\end{array}\right],
\label{eqn_platoonsys_space_x_F}
\end{align}
whereas the use of (\ref{eqn_delta1}) and (\ref{eqn_output_yi}) yields the output equation
\begin{align}
H(x_i) = (1-\kappa_0)\Delta_i - \delta_{1,i}. \label{eqn_platoonsys_space_x_H}
\end{align}
Finally, the rows $\bar{\rho}_{k}$ in the matrix-valued function $\bar{\rho}$ in (\ref{eqn_platoonsys_space_x}) can be obtained through the dynamics for $\delta_i$ and the definition (\ref{eqn_rhok}), leading to
\begin{align}
\bar{\rho}_{1}(\xi_i,\xi_{i-1},\xi_0) &= \left[\begin{array}{c}
\kappa\rho_1^{\T}(\xi_i) \\ 0  \\ 0 \end{array}\right]^{\T},
\label{eqn_rhobar1}\\
\bar{\rho}_{k}(\xi_i,\xi_{i-1},\xi_0) &= \left[\begin{array}{c}
\kappa \rho_{k}^{\T}(\xi_i) + \rho_{k-1}^{\T}(\xi_i) \\
(\kappa_0-1)\rho_{k-1}^{\T}(\xi_{i-1}) \\
-\kappa_0\rho_{k-1}^{\T}(\xi_0) \end{array}\right]^{\T}, \label{eqn_rhobark}
\end{align}
with $k = 2,\ldots,n-1$.
Here, it can be observed that the definition of the spacing policy $\delta_{1,i}$ in (\ref{eqn_delta1}) implies that the disturbances on both the preceding vehicle and first vehicle in the platoon affect the timing error of vehicle $i$. As a result, $\bar{w}_i$ in (\ref{eqn_platoonsys_space_x}) is defined as $\bar{w}_i = [\begin{array}{ccc} w_i^{\T} & w_{i-1}^{\T} & w_0^{\T} \end{array}]^{\T}$.
\begin{remark}\label{rem_leadvehicle}
The timing errors $\Delta_i$ and $\Delta_i^0$ in (\ref{eqn_Delta}) and (\ref{eqn_Delta0}), respectively, as well as $\delta_i$ in (\ref{eqn_delta1}), (\ref{eqn_deltak}) are not defined for the lead vehicle with index $i=0$. Instead, take $\Delta_0 \defl t_0 - \int\vref^{-1}\dd s$ as the deviation from a nominal trajectory and let $\Delta_0^0 \defl \Delta_0$. Then, $\delta_{1,0}$ can be defined according to~(\ref{eqn_delta1}) as $\delta_{1,0} = \Delta_0(s) + \kappa e_{1,0}$. Similarly, $\delta_{k,0} = e_{k-1,0} + \kappa e_{k,0}$. It then follows from the dynamics~(\ref{eqn_platoonsys_space_te}) that the first vehicle in the platoon satisfies
\begin{align}
\begin{split}
\mathring{x}_0(s) &= F\big(x_0(s),\tilde{u}_0(s),0,\bar{\rho}(\xi_0,0,0)\bar{w}_0(s)\big), \\
y_0(s) &= H(x_0(s)),
\end{split}\label{eqn_platoonsys_space_x_i=0}
\end{align}
with $F$ and $H$ as in (\ref{eqn_platoonsys_space_x_F}) and (\ref{eqn_platoonsys_space_x_H}), respectively. In (\ref{eqn_platoonsys_space_x_i=0}), (\ref{eqn_output_yi}) is used for $i=0$ to obtain the latter equation and $\bar{w}_0 = [\begin{array}{ccc} w_0^{\T} & 0 & 0\end{array}]^{\T}$. It is clear that the dynamics (\ref{eqn_platoonsys_space_x_i=0}) is of the same form as that of the follower vehicles in (\ref{eqn_platoonsys_space_x}).
\end{remark}

Since the objective is to achieve the desired spacing policy for vehicle $i$ by the design of a controller that stabilizes $\delta_{1,i}=0$, the subspace $\Scal_i$ is introduced as
\begin{align}
\Scal_i \defl \big\{ x\in\R^{(N+1)n} \;\big|\; \delta_{i} = 0 \big\}, \;\; i\in\Ical_N^0.
\label{eqn_Scal}
\end{align}
Here, $x = [\begin{array}{cccc} x_0^{\T} & x_1^{\T} & \cdots & x_N^{\T}\end{array}]^{\T}$ is the state of the platoon, where it is recalled that $x_i = [\begin{array}{cc} \Delta_i & \delta_i^{\T} \end{array}]^{\T}$ satisfies the dynamics (\ref{eqn_platoonsys_space_x})--(\ref{eqn_platoonsys_space_x_H}). In order to render $\Scal_i$ positively invariant in the absence of disturbances, a controller $\tilde{u}_i = k(x_i)$ (i.e., a decentralized controller) is sought that achieves input-to-state stability with respect to the set $\Scal_i$ as in (\ref{eqn_Scal}), i.e., there exist functions $\beta_{\delta}$ of class $\classKL$ and $\sigma_{\delta}$ of class $\classKinf$ such that the controlled system satisfies
\begin{align}
|x(s)|_{\Scal_i} \leq \beta_{\delta}\big(|x(0)|_{\Scal_i},s-s_0\big) + \sigma_{\delta}\!\left(\|\bar{w}_i\|_{\infty}^{[s_0,s]}\right),
\label{eqn_lem_control_setiss}
\end{align}
where $|x|_{\Scal_i} \defl \inf_{z\in\Scal_i}|x - z|$ represents the distance to $\Scal_i$.

After introducing the set $\Xcal_c\ni0$ as
\begin{align}
\textstyle
\Xcal_{c} \defl \big\{ x\in\R^{(N+1)n} \;\big|\; \sup_{j\in\Ical_N^0}|x_j| \leq c \big\},
\label{eqn_Xcal}
\end{align}
it can be shown that a controller exists that locally achieves~(\ref{eqn_lem_control_setiss}).
\begin{lemma}\label{lem_control}
Consider the platoon dynamics (\ref{eqn_platoonsys_space_x})--(\ref{eqn_platoonsys_space_x_H}) resulting from the vehicle dynamics (\ref{eqn_platoonsys_space}) and the spacing policy (\ref{eqn_delta1}) and let $\tilde{u}_i = K\delta_i$ be a feedback controller in which $K$ is chosen such that the matrix $A + \kappa BK$ is Hurwitz, with $A$ and $B$ as in (\ref{eqn_platoonsys_space_matrices}). Then, there exists a constant $\bar{c}_{\delta}>0$ such that, for any trajectory $x(\cdot)$ that satisfies $x(s)\in\Xcal_{\bar{c}_{\delta}}$ for all $s\geq0$, (\ref{eqn_lem_control_setiss}) holds.
\end{lemma}
\begin{proof}
The proof can be found in Appendix~\ref{app_lem_control_proof}.
\end{proof}
\begin{remark}
In the absence of disturbances, the condition (\ref{eqn_lem_control_setiss}) implies that the set $\Scal_i$ is controlled invariant. Also, it is remarked that the invariance of $\Scal_i$ is independent of the control input for other vehicles (with index different from $i$), which is the result of the choice of $\delta_{1,i}$ in (\ref{eqn_delta1}). This choice also directly determines the dynamics on the invariant set $\Scal_i$, which is given by the first equation in (\ref{eqn_platoonsys_space_x}).
\end{remark}
\begin{remark}
The controller for vehicle $i\in\Ical_N$ given by (\ref{eqn_inoutlin_ui}), (\ref{eqn_input_ubar}), and $\tilde{u}_i = K\delta_i$ as obtained through Lemma~\ref{lem_control} relies on state information from the preceding vehicle (with index $i-1$) and, if $\kappa_0>0$, from the lead vehicle with index $0$. In particular, this information is required for a given position~$s$. As the lead vehicle and preceding vehicle pass some time before vehicle~$i$, this control approach is inherently robust to (small) time-delays in wireless communication, which is typically used to share this information.
\end{remark}
\begin{remark}
Even though the controller designed in this section is specified in the spatial domain, that does not prohibit the practical implementation of such controller in the time domain. To illustrate this, consider the computation of the timing error $\Delta_i(s)$ in (\ref{eqn_Delta}). Consider a vehicle $i$ and let $s_i(t)$ be its current position (specified in the time domain). Similarly, let $s_{i-1}(\cdot)$ be the historical evolution of the position of the preceding vehicle that can be obtained through wireless communication or from radar measurements and of which a sampled version can be stored onboard vehicle $i$ with limited memory. Then, the timing error $\Delta_i(s)$ (for $s = s_i(t)$) can be obtained by (numerically) solving the implicit equation $s_{i-1}(t - \Delta t + \Delta_i(s)) = s_{i}(t)$. 
Next, if $\tilde{\xi}_i(t)$ represents the current state of vehicle $i$ (again specified in the time domain), then $\xi_i(s) = \tilde{\xi}_i(t)$ with $s = s_i(t)$ and the state $\xi_{i-1}(s)$ of the preceding vehicle at the same point in space can be obtained from a time-domain specification $\tilde{\xi}_{i-1}(\cdot)$ of its state as $\xi_{i-1}(s) = \tilde{\xi}_{i-1}(t - \Delta t + \Delta_i(s))$. Finally, it is remarked that the controller synthesis procedure of this section is constructive, allowing, in principle, for practical implementation. Future work will focus on this aspect.
\end{remark}

\subsection{Platoon disturbance string stability analysis}\label{sec_platoonstability}
The application of any controller that achieves (\ref{eqn_lem_control_setiss}) leads to a controlled platoon that is disturbance string stable when leader information is exploited, i.e., when $\kappa_0>0$ in (\ref{eqn_delta1}). This is formalized for feedback controllers of the form $\tilde{u}_i = K\delta_i$ in the following theorem.
\begin{theorem}\label{thm_diststringstab}
Consider the platoon dynamics (\ref{eqn_platoonsys_space_x})--(\ref{eqn_platoonsys_space_x_H}) resulting from the vehicle dynamics (\ref{eqn_platoonsys_space}) and the spacing policy (\ref{eqn_delta1}) and let $\tilde{u}_i = K\delta_i$, $i\in\Ical_N^0$, be a feedback controller in which $K$ is chosen such that the matrix $A + \kappa BK$ is Hurwitz, with $A$ and $B$ as in (\ref{eqn_platoonsys_space_matrices}). Then, the closed-loop platoon system is disturbance string stable if $\kappa_0 > 0$.
\end{theorem}
\begin{proof}
The proof is given in Appendix~\ref{app_thm_diststringstab_proof}.
\end{proof}

The proof of Theorem~\ref{thm_diststringstab} shows that (\ref{eqn_thm_localios}) holds with $\gamma_y(r) = (1-\kappa_0)r$ and then employs Theorem~\ref{thm_localios} to guarantee disturbance string stability. As the the function $\gamma_y$ is only dependent on $\kappa_0$, it is clear that the result in Theorem~\ref{thm_diststringstab} is independent of the specific controller design. Instead, the result holds for any controller that achieves (\ref{eqn_lem_control_setiss}) (i.e., that renders $\Scal_i$ controlled invariant), indicating that the disturbance string stability property is a result of the choice of the spacing policy $\delta_{1,i}$ in (\ref{eqn_delta1}) rather than the specific controller. It is also noted that Theorem~\ref{thm_diststringstab} shows local disturbance string stability. Global string stability can not be shown as the vehicle velocities $v_i = \tilde{h}(\xi_i)$ need to be strictly positive to ensure that the dynamics (\ref{eqn_platoonsys_space}) in the spatial domain is well-defined, see (\ref{eqn_platoonsys_space_functions}).
\begin{remark}\label{rem_timedomain}
The controller design achieving disturbance string stability discussed in Lemma~\ref{lem_control} and Theorem~\ref{thm_diststringstab} is done in the spatial domain in order to obtain a delay-independent analysis of the delay-based spacing policy (\ref{eqn_constanttimegap}). Nonetheless, the results obtained in this section are directly applicable to other spacing policies when the vehicle dynamics (\ref{eqn_platoonsys_time}) is considered in time domain. Namely, the spacing errors can be defined as
\begin{align}
\Delta_i(t) &\defl s_i(t) - s_{i-1}(t) + d, \nonumber\\
\Delta_i^0(t) &\defl s_i(t) - s_0(t) + id, \quad i\in\Ical_N, \label{eqn_Delta_time}
\end{align}
providing counterparts of (\ref{eqn_Delta}) and (\ref{eqn_Delta0}). Then, after defining the velocity tracking error as
\begin{align}
e_{1,i}(t) \defl \tilde{h}(\xi_i(t)) - \vref(t)
\label{eqn_veltrackerror_e1_time}
\end{align}
for a reference velocity $\vref$ (specified in time domain) and the introduction of the spacing policy
\begin{align}
\delta_{1,i}(t) \defl (1-\kappa_0)\Delta_i(t) + \kappa_0\Delta_i^0(t) + \kappa e_{1,i}(t),
\label{eqn_delta1_time}
\end{align}
the results of Lemma~\ref{lem_control} and Theorem~\ref{thm_diststringstab} directly hold. Also, it is noted that (\ref{eqn_delta1_time}) represents the constant headway spacing policy (\ref{eqn_constantheadway}) for $\kappa_0=0$.
\end{remark}

In the absence of disturbance, i.e., $w_i=0$ for all $i\in\Ical_N^0$, the sets $\Scal_i$ in (\ref{eqn_Scal}) are positively invariant, as follows from (\ref{eqn_lem_control_setiss}) and the controller design in Lemma~\ref{lem_control}. Consequently, the set $\Scal \defl \bigcap_{i\in\Ical_N^0}\Scal_i$ is positively invariant as well. The dynamics on $\Scal$ is given by
\begin{align}
\begin{split}
\kappa\mathring{\Delta}_0(s) &= -\Delta_0(s), \\
\kappa\mathring{\Delta}_i(s) &= -\Delta_i(s) + (1-\kappa_0)\Delta_{i-1}(s), \quad i\in\Ical_N,
\end{split}\label{eqn_platoonsys_space_delta_Scal}
\end{align}
as follows from (\ref{eqn_platoonsys_space_x}) and (\ref{eqn_platoonsys_space_x_i=0}) for $\delta_i = 0$ (i.e., on the invariant set $\Scal$). It is important to note that (\ref{eqn_platoonsys_space_delta_Scal}) is a direct consequence of the choice of the spacing policy (\ref{eqn_delta1}) rather than the details of the designed controller. Then, on the set $\Scal$, perturbations on the timing error $\Delta_i$ (e.g., when $\Delta_0(0) \neq 0$) do not grow as they propagate through the string, as formally stated as follows.
\begin{proposition}\label{lem_L2stringstab}
Consider the dynamics (\ref{eqn_platoonsys_space_delta_Scal}) and initial conditions satisfying $\Delta_i(0) = 0$ for all $i\in\Ical_N$. Then, for all $i\in\Ical_N$, the timing errors $\Delta_i$ satisfy
\begin{align}
\!\!\int_0^s |\Delta_i(\theta)|^2 \di\theta \leq (1-\kappa_0)^2\int_0^s |\Delta_{i-1}(\theta)|^2 \di\theta, \; \forall s\geq0.\!\!
\label{eqn_lem_L2stringstab}
\end{align}
\end{proposition}
\begin{proof}
Let $i\in\Ical_N$ and define the function $V(\Delta_i) \defl \tfrac{1}{2}\kappa\Delta_i^2$. Then, the differentiation of $V$ with respect to space and along trajectories of (\ref{eqn_platoonsys_space_delta_Scal}) yields
\begin{align}
\mathring{V}(\Delta_i) &= -\Delta_i^2 + (1-\kappa_0)\Delta_i\Delta_{i-1}, \\
&= -\tfrac{1}{2}|\Delta_i|^2 + \tfrac{1}{2}|(1-\kappa_0)\Delta_{i-1}|^2 \nonumber\\
&\qquad- \tfrac{1}{2}|(1-\kappa_0)\Delta_{i-1} - \Delta_i|^2, \label{eqn_lem_L2stringstab_proof_step1}
\end{align}
where the latter equality can be checked by completing the squares. The integration of (\ref{eqn_lem_L2stringstab_proof_step1}), hereby recalling that $\Delta_i(0) = 0$ and noting $V(\Delta_i(s))\geq0$, leads to the result (\ref{eqn_lem_L2stringstab}).
\end{proof}
It is remarked that (\ref{eqn_lem_L2stringstab}) essentially represents a string stability property using $\Lcal_2$ signal norms \cite{book_vanderschaft_2000}, albeit with space as the independent variable. Contrary to the case of disturbance string stability in Theorem~\ref{thm_diststringstab}, the stability property in Proposition~\ref{lem_L2stringstab} guarantees that perturbations do not grow unbounded even in the case of absence of leader information (i.e., $\kappa_0=0$).

\begin{table}
\begin{center}
  \caption{\rm Parameter values for the example considered in Section~\ref{sec_evaluation}.}
  \label{tab_example_parameters}
  \begin{tabular}{|c|c||c|c|}
  \hline
  parameter & value & parameter & value \\
  \hline
  $\tau$ & $1$ & $\omega_0$ & $0.05$\\
  $\dt$  & $1$ & $\zeta_0$ & $0.9$\\
  $\kappa_0$ & $0.1$ & $K_1$ & $2\zeta_0\omega_0$\\
  $\kappa$ & $2$ & $K_2$ & $\omega_0^2$ \\
  \hline
  \end{tabular}
\end{center}
\end{table}

\section{Evaluation}\label{sec_evaluation}
In order to evaluate the platooning controller design procedure discussed in Section~\ref{sec_platooncontrol}, the vehicle dynamics (in time domain)
\begin{align}
\begin{split}
\dot{s}_i(t) &= v_i(t), \\
\dot{v}_i(t) &= a_i(t) + w_i(t), \\
\tau\dot{a}_i(t) &= -a_i(t) + u_i(t),
\end{split}\label{eqn_example_dynamics_time}
\end{align}
is considered. Here, $s_i$, $v_i$, and $a_i$ represent the vehicle position, velocity, and acceleration, respectively. It is easily seen that (\ref{eqn_example_dynamics_time}) is of the form (\ref{eqn_platoonsys_time}) with $\xi_i = [\begin{array}{cc} v_i & a_i \end{array}]^{\T}$ and satisfies Assumption~\ref{ass_relativedegree}. The model (\ref{eqn_example_dynamics_time}) extends the vehicle model considered in, e.g., \cite{stankovic_2000,ploeg_2014b}, by the inclusion of external disturbance~$w_i$. It is noted that the approach introduced in this paper allows for more general nonlinear models that include, e.g., aerodynamic effects and engine dynamics. Examples of such detailed vehicle models can be found in~\cite{book_kiencke_2005,book_ulsoy_2012}.

Following the motivation in Section~\ref{sec_spacingpolicies}, a delay-based spacing policy is considered and a controller according to Section~\ref{sec_platooncontrol} is synthesized in the spatial domain. To this end, it is noted that the velocity tracking errors (\ref{eqn_veltrackerror_e1}), (\ref{eqn_veltrackerror_ek}) for the dynamics (\ref{eqn_example_dynamics_time}) in spatial domain read
\begin{align}
e_{1,i}(s) &=  \frac{1}{v_i(s)}-\frac{1}{\vref(s)}, \\
e_{2,i}(s) &= -\frac{a_i(s)}{v_i^3(s)}-\frac{\dd}{\dd s}\!\left(\frac{1}{\vref(s)}\right).\label{eqn_example_e2i}
\end{align}
Here, it is noted that (\ref{eqn_example_e2i}) is related to the acceleration of the vehicle, albeit expressed in the spatial domain. Then, the feedback linearizing controller in (\ref{eqn_inoutlin_ui}) is given as
\begin{align}
u_i(s) &= a_i(s) + 3\tau\frac{a_i^2(s)}{v_i(s)} \nonumber\\
&\qquad- \tau v_i^4(s)\left( \frac{\dd^2}{\dd s^2}\!\left(\frac{1}{\vref(s)}\right) + \bar{u}_i(s)\right),
\end{align}
after which (\ref{eqn_input_ubar}) and the feedback $\tilde{u}_i = K\delta_i$ in Lemma~\ref{lem_control} read
\begin{align}
\bar{u}_i(s) &= \frac{1}{\kappa}\frac{a_i(s)}{v_i^3(s)}
- \frac{1-\kappa_0}{\kappa}\frac{a_{i-1}(s)}{v_{i-1}^3(s)}
- \frac{\kappa_0}{\kappa}\frac{a_0(s)}{v_0^3(s)} \nonumber\\
&\qquad+ K_1\delta_{1,i}(s) + K_2\delta_{2,i}(s),
\end{align}
with $\delta_i$ as in (\ref{eqn_delta1}) and (\ref{eqn_deltak}). The nominal parameters of the vehicle model (\ref{eqn_example_dynamics_time}), spacing policy (\ref{eqn_delta1}), and controller are given in Table~\ref{tab_example_parameters}. As $\kappa_0>0$, vehicles exploit information from both their predecessor and the platoon leader.

\begin{figure}
\begin{center}
  \includegraphics[scale=1]{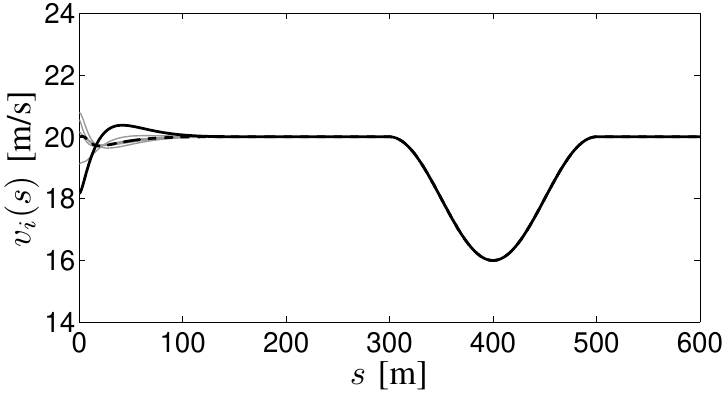}
  \caption{Velocities $v_i$ for the lead vehicle (black) and $N=5$ follower vehicles (gray, with the last one in dashed black) for the delay-based policy (\ref{eqn_delta1}). The initial conditions are randomly generated, the reference velocity reads $\vref(s) = 20 - 2(1-\cos(10^{-2}\pi(s-300)))$\,m/s for $300\leq s\leq500$ and $\vref(s) = 20$\,m/s otherwise.}
  \label{fig_sim_ic_vref_space_vel}
\end{center}
\end{figure}

\begin{figure}
\begin{center}
  \includegraphics[scale=1]{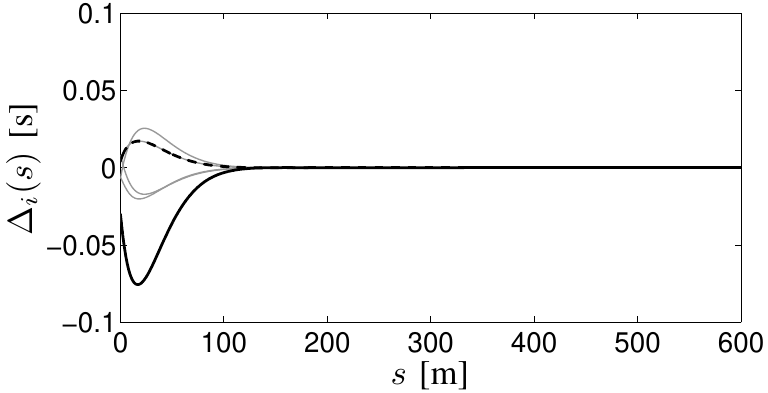}  
  \caption{Timing errors $\Delta_i$ as in (\ref{eqn_Delta}) for the first follower vehicle (black) and the remaining follower vehicles (gray) corresponding to the case in Figure~\ref{fig_sim_ic_vref_space_vel}.}
  \label{fig_sim_ic_vref_space_Delta}
\end{center}
\end{figure}

\begin{figure}
\begin{center}
  \includegraphics[scale=1]{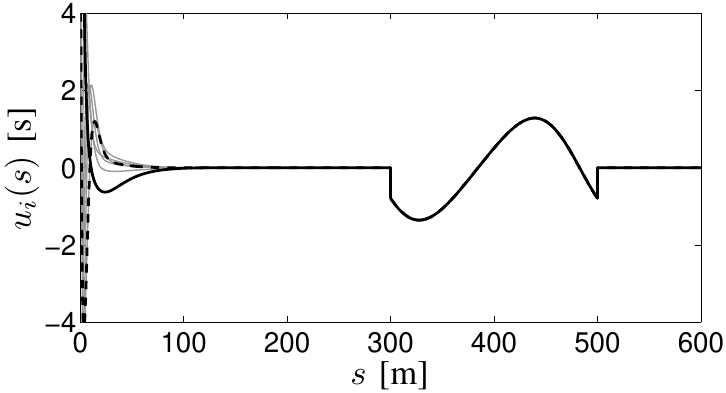}  
  \caption{Control inputs $u_i$ for the lead vehicle (black) and the follower vehicles (gray) corresponding to the case in Figure~\ref{fig_sim_ic_vref_space_vel}.}
  \label{fig_sim_ic_vref_space_input}
\end{center}
\end{figure}

The tracking of a reference velocity profile $\vref$ satisfying Assumption~\ref{ass_vref} is considered in Figures~\ref{fig_sim_ic_vref_space_vel} to~\ref{fig_sim_ic_vref_space_input}, where initial conditions are randomly generated perturbations of the equilibrium. By control design, this equilibrium satisfies $\Delta_i=0$ and $\delta_i=0$ for all $\Ical_N^0$, as can be observed in (\ref{eqn_platoonsys_space_x}) (see also (\ref{eqn_platoonsys_space_x_i=0}) in Remark~\ref{rem_leadvehicle} for the lead vehicle). By the definitions of $\Delta_i$ in (\ref{eqn_Delta}) and $\delta_{i,1}$ in (\ref{eqn_delta1}), it follows that the velocity error $e_{i,1}$ satisfies $e_{i,1} = 0$ at this equilibrium, such that the desired velocity profile is tracked. In the absence of disturbances (i.e., $w_i=0$), it follows from Theorem~\ref{thm_diststringstab} that this equilibrium is asymptotically stable. This is also observed by the tracking of the reference velocity in Figure~\ref{fig_sim_ic_vref_space_vel}, whereas Figure~\ref{fig_sim_ic_vref_space_Delta} shows that the desired spacing policy is obtained. It is recalled that these objectives are compatible through Proposition~\ref{lem_timegap}. Finally, it is clear from the input signals in Figure~\ref{fig_sim_ic_vref_space_input} that all vehicles have the same behavior in the spatial domain.

In order to illustrate that the tracking of a spatially varying reference velocity is a distinguishing feature of the delay-based spacing policy, the constant headway policy is considered as an alternative. In particular, the spacing policy (\ref{eqn_delta1_time}) in Remark~\ref{rem_timedomain} is considered, even though (\ref{eqn_veltrackerror_e1_time}) is replaced by $e_{1,i}(t) = v_i(t) - \vref(s_i(t))$ to target tracking of the spatially varying reference velocity. For this spacing policy, a controller is synthesized in time domain according to the discussion in Remark~\ref{rem_timedomain}, where the parameters in Table~\ref{tab_example_parameters} are adapted to give the same time scales as the controller used in Figures~\ref{fig_sim_ic_vref_space_vel} and~\ref{fig_sim_ic_vref_space_Delta} for a nominal velocity of $v_{\text{\rm nom}} = 20$\,m/s.

\begin{figure}
\begin{center}
  \includegraphics[scale=1]{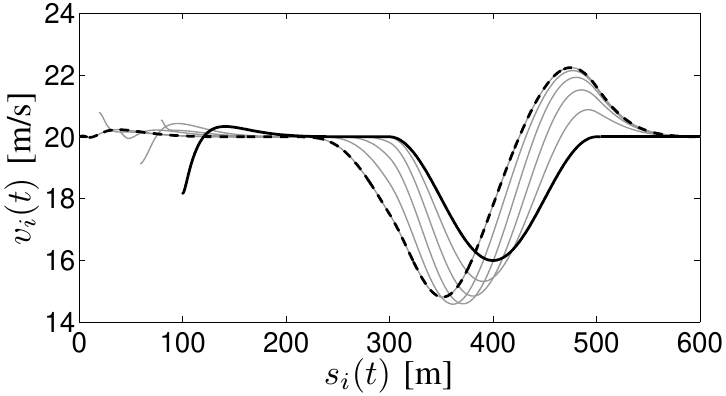}  
  \caption{Velocities $v_i$ for the lead vehicle (black) and $N=5$ follower vehicles (gray) for the constant headway gap policy (\ref{eqn_delta1_time}). The same case as in Figure~\ref{fig_sim_ic_vref_space_vel} is considered. Moreover, the parameter values for (\ref{eqn_delta1_time}) and the controller are chosen such that the resulting time scales equal that of the controller for the delay-based policy in Figure~\ref{fig_sim_ic_vref_space_vel} when a nominal velocity of $v_{\rm nom} = 20$\,m/s is used.}
  \label{fig_sim_ic_vref_time_vel}
\end{center}
\end{figure}

\begin{figure}
\begin{center}
  \includegraphics[scale=1]{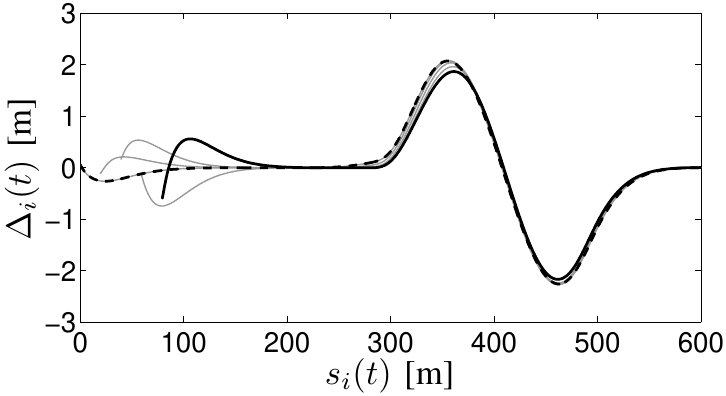}  
  \caption{Spacing errors $\Delta_i$ as in (\ref{eqn_Delta_time}) for the first follower vehicle (black) and the remaining follower vehicles (gray) corresponding to the case in Figure~\ref{fig_sim_ic_vref_time_vel}.}
  \label{fig_sim_ic_vref_time_Delta}
\end{center}
\end{figure}

\begin{figure}
\begin{center}
  \includegraphics[scale=1]{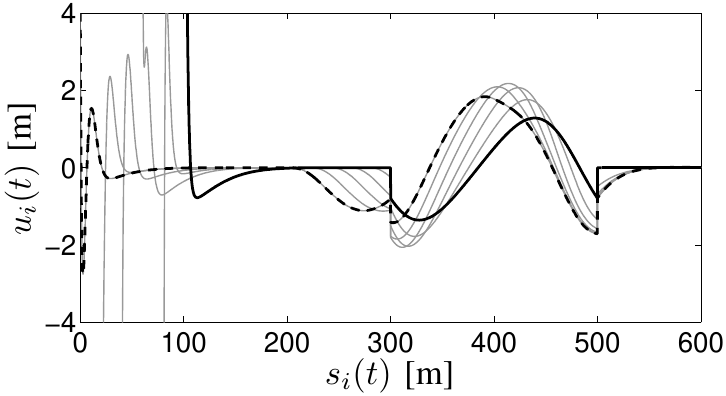}  
  \caption{Control inputs $u_i$ for the lead vehicle (black) and the follower vehicles (gray) corresponding to the case in Figure~\ref{fig_sim_ic_vref_time_vel}.}
  \label{fig_sim_ic_vref_time_input}
\end{center}
\end{figure}

The results of this time domain controller using a constant headway strategy are depicted in Figures~\ref{fig_sim_ic_vref_time_vel} to~\ref{fig_sim_ic_vref_time_input}. It can be observed that this controller indeed achieves the stabilization of the desired equilibrium point as long as the reference velocity is constant. However, it is clear from Figure~\ref{fig_sim_ic_vref_time_vel} that the vehicles do not accurately track the desired reference velocity (defined in the spatial domain). Moreover, the change in reference velocity leads to a perturbation in the achieved spacing as well, as depicted in Figure~\ref{fig_sim_ic_vref_time_Delta}, with the control inputs in Figure~\ref{fig_sim_ic_vref_time_input}. Even though emphasis can be put on either the tracking of the reference velocity or the desired spacing through the choice of the parameter $\kappa$ in (\ref{eqn_delta1_time}), it is stressed that an increase in tracking performance of the reference velocity will lead to a less accurate tracking of the spacing policy (and vice versa). Namely, the tracking of a spatially varying reference velocity is fundamentally incompatible with the simultaneous tracking of a constant headway policy, as discussed in Section~\ref{sec_spacingpolicies}. Consequently, the use of alternative control strategies will not mitigate this effect.

\begin{figure}
\begin{center}
  \includegraphics[scale=1]{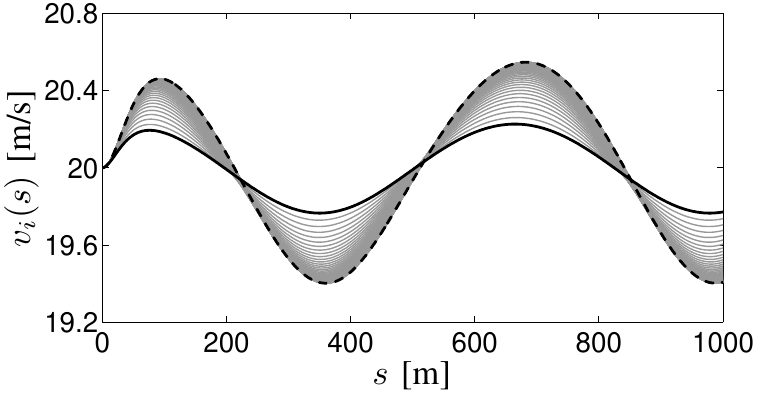}  
  \caption{Velocities $v_i$ for the lead vehicle (black) and $N=50$ follower vehicles (gray, with the last one in dashed black) for zero initial conditions and $\vref(s) = 20$\,m/s for all $s\geq0$. The disturbance is given as $w_i(s) = \sin(10^{-2}s)$ for all $i\in\Ical_N^0$.}
  \label{fig_sim_disturbance}
\end{center}
\end{figure}

\begin{figure}
\begin{center}
  \includegraphics[scale=1]{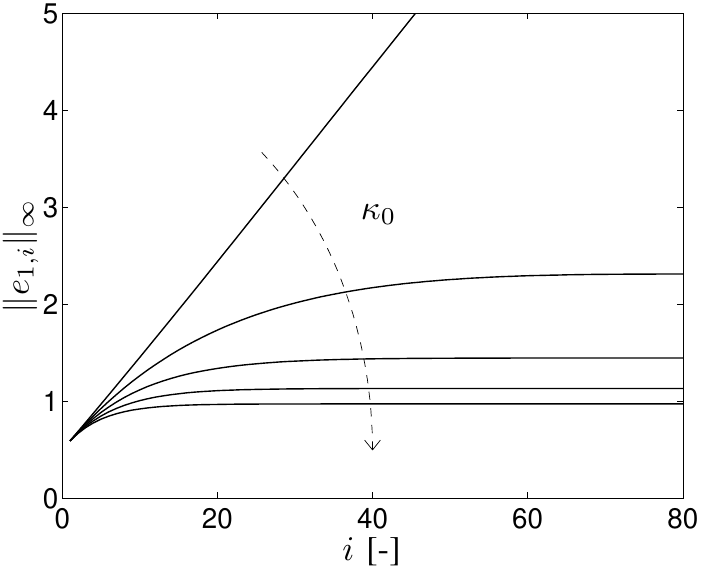}  
  \caption{Maximum velocity errors $e_{1,i}$ as in (\ref{eqn_veltrackerror_e1}) for $\kappa_0 \in\{0,0.05,0.1,0.15,0.2\}$ and disturbance $w_i(s) = \sin(10^{-2}s)$ for $i = \Ical_N$ and $w_0(s) = 0$. As $\kappa_0$ grows, $\|e_{1,i}\|_{\infty}$ decreases, as indicated by the dashed arrow.}
  \label{fig_disturbance_eLinf_N=102}
\end{center}
\end{figure}

Returning to the case of the delay-based spacing policy, Figure~\ref{fig_sim_disturbance} shows the velocities $v_i$ of $N+1=51$ vehicles subject to a common disturbance, hereby again using the parameter values in Table~\ref{tab_example_parameters}. As the disturbance is bounded, the results of Theorem~\ref{thm_diststringstab} hold and the platoon is disturbance string stable as in Definition~\ref{def_diststringstab}. Consequently, there is a uniform (over the platoon index) bound on the deviations from the equilibrium (given by $\Delta_i = 0$ and $\delta_i=0$, for which $v_i(s) = \vref(s)$), as can also be observed in Figure~\ref{fig_sim_disturbance}. Next, the maximum velocity errors $e_{1,i}$ for a platoon with $N = 80$ follower vehicles are depicted in Figure~\ref{fig_disturbance_eLinf_N=102} for varying values of $\kappa_0$. Here, the same disturbance as in Figure~\ref{fig_sim_disturbance} is considered. As stated in Theorem~\ref{thm_diststringstab}, there are uniform bounds on the velocity errors when $\kappa_0 > 0$, i.e., when information of the lead vehicle is shared with all other vehicles in the platoon. It is noted that the case $\kappa_0=0$ indeed leads to unbounded velocity errors for growing platoon size (i.e., an absence of disturbance string stability). This indicates that the results of Theorems~\ref{thm_localios} and~\ref{thm_diststringstab} are not conservative.

\section{Conclusions}\label{sec_conclusions}
The control of vehicle platoons was considered in this paper, hereby exploiting a novel delay-based spacing policy that guarantees that all vehicles in the platoon track the same velocity profile in the spatial domain. This property is particularly relevant for vehicles that track a spatially varying velocity profile, such as heavy-duty vehicles driving over hilly terrain. The influence of external disturbances was addressed by the introduction of disturbance string stability. A controller was designed that tracks a reference velocity profile, maintains the desired spacing policy, and achieves disturbance string stability. In fact, it was shown that string stability is the result of the spacing policy rather than the specific controller design.

Even though homogeneous vehicle platoons were considered, the controller design approach presented in this paper has the potential to be applicable to heterogeneous vehicle platoon as well. Apart from being supported by the notion of disturbance string stability, the result that the details of controller design are less crucial than the chosen spacing policy suggest that non-identical vehicles can be considered as long as a common spacing policy is adopted.
Also, it is remarked that the space-based control approach taken in this paper can be particularly relevant for the lateral control of vehicles in a platoon, as road features such as corners are specified in the spatial domain rather than time domain. Future work will focus on these aspects as well as on the practical implementation of controllers designed in the spatial domain (and potentially in the presence of measurement errors). Another interesting direction for future research is the extension of the controller design to vehicle platoons with general interconnection topology, as can be enabled by wireless communication.

\appendix
\section{Proofs}
\subsection{Proof of Theorem~\ref{thm_localiss}}\label{app_thm_localiss_proof}
The theorem will be proven in two steps. First, it will be shown that all $x_i(\theta)$ are bounded for all $\theta\geq\theta_0$ and, second, the bound of the form (\ref{eqn_def_diststringstab}) will be shown.

In order to prove boundedness of $x_i(\theta)$, constants $\bar{c}$ and $\bar{c}_w$ satisfying  $0<\bar{c} < c$ and $0<\bar{c}_w < c_w$ are introduced. A specific choice for $\bar{c}$ and $\bar{c}_w$ will be made later. Now, taking initial conditions $|x_i(\theta_0)| < \bar{c}$ and disturbances $w_i$ bounded as $\infnorm{w_i}{[\theta_0,\infty)} < \bar{c}_w$, it is noted that (\ref{eqn_thm_localiss}) gives
\begin{align}
	\infnorm{x_i}{[\theta_0,\theta]} &\leq \beta\big(|x_i(\theta_0)|,0\big) \nonumber\\
	&\qquad + \bar{\gamma}\infnorm{x_{i-1}}{[\theta_0,\theta]} + \sigma\big(\infnorm{w_i}{[\theta_0,\theta]}\big),
	\label{eqn_thm_localiss_proof_step1}
\end{align}
for all $i\in\Ical_N$ whenever $\|x_{i-1}\|_{\infty}<c$. Moreover, it is noted that, due to the structure of the interconnection in (\ref{eqn_cascsys_dist}), the bound for system $i=0$ reads
\begin{align}
\infnorm{x_0}{[\theta_0,\theta]} \leq \beta\big(|x_0(\theta_0)|,0\big) + \sigma\big(\infnorm{w_0}{[\theta_0,\theta]}\big).
\label{eqn_thm_localiss_proof_step2}
\end{align}
Then, it can be concluded that the recursive application of (\ref{eqn_thm_localiss_proof_step1}) and the use of (\ref{eqn_thm_localiss_proof_step2}) yields
\begin{align}
\infnorm{x_i}{[\theta_0,\theta]} &\leq \sum_{j=0}^i \bar{\gamma}^{i-j} \beta\big(|x_j(\theta_0)|,0\big) \nonumber\\
&\qquad + \sum_{j=0}^i \bar{\gamma}^{i-j} \sigma\big(\infnorm{w_j}{[\theta_0,\theta]}\big),\!\!
\label{eqn_thm_localiss_proof_step3}
\end{align}
for all $i\in\Ical_N^0$. By the properties of the class \classKL function $\beta$, it directly follows that $\beta(|x_j(0)|,0)\leq\beta(\sup_{k\in\Ical_N^0}|x_k(0)|,0)$ for any $j\in\Ical_N^0$, so that a uniform bound is obtained on all terms that depend on the initial condition. A similar bound can be obtained on $\sigma(\infnorm{w_j}{[\theta_0,\theta]})$. In addition,
\begin{align}
\sum_{j=0}^i \bar{\gamma}^{i-j}
\leq  \sum_{j=0}^N \bar{\gamma}^{N-j}
<  \sum_{l=0}^\infty \bar{\gamma}^l = \frac{1}{1-\bar{\gamma}},
\label{eqn_thm_localiss_proof_step4}
\end{align}
which follows from noting that the sum in (\ref{eqn_thm_localiss_proof_step4}) represents a geometric series with $0<\bar{\gamma}<1$. The use of these bounds in (\ref{eqn_thm_localiss_proof_step3}) gives
\begin{align}
\infnorm{x_i}{[\theta_0,\theta]} &\leq \frac{1}{1-\bar{\gamma}} \, \beta\!\left(\sup_{j\in\Ical_N^0}|x_j(\theta_0)|,0\right) \nonumber\\
&\qquad+ \frac{1}{1-\bar{\gamma}} \, \sigma\!\left(\sup_{j\in\Ical_N^0}\infnorm{w_j}{[\theta_0,\theta]}\right),
\label{eqn_thm_localiss_proof_step5}
\end{align}
for all $i\in\Ical_N^0$ and all $N\in\N$. Here, it is stressed that (\ref{eqn_thm_localiss_proof_step5}) represents a uniform bound on state perturbations for all systems in a possibly (countably) infinite interconnection.
Moreover, when the constants $\bar{c} < c$ and $\bar{c}_w < c_w$ are taken to satisfy $\beta(\bar{c},0) + \sigma(\bar{c}_w) < (1-\bar{\gamma})c$, it is clear that $\|x_i\|_{\infty} < c$ for all $i\in\Ical_N^0$ and the derivation above is consistent with the assumptions in the statement of the theorem.

For future reference, the function $\delta$ and constant $\Delta^{[\theta_0,\theta]}$ are introduced as
\begin{align}
\delta(r) &\defl \frac{1}{1-\bar{\gamma}}\,\beta(r,0), \label{eqn_thm_localiss_proof_delta}\\
\Delta^{[\theta_0,\theta]} &\defl \frac{1}{1-\bar{\gamma}}\,\sigma\!\left(\sup_{j\in\Ical_N^0}\infnorm{w_i}{[\theta_0,\theta]}\right),\!
\label{eqn_thm_localiss_proof_Delta}
\end{align}
such that (\ref{eqn_thm_localiss_proof_step5}) can be written as $\infnorm{x_i}{[\theta_0,\theta]} \leq \delta(\sup_{j\in\Ical_N^0}|x_j(\theta_0)|) + \Delta^{[\theta_0,\theta]}$. Note that the function $\delta$ is of class $\classK$.

It remains to be proven that there exists an estimate of the form (\ref{eqn_def_diststringstab}), in which the influence of the initial condition vanishes as $\theta\rightarrow\infty$. To this end, consider system $i\in\Ical_N^0$ and let $\{\vartheta_j^i\}_{j=0}^{i+1}$ be a sequence that satisfies
\begin{align}
\theta_0 < \vartheta_0^i < \vartheta_1^i < \ldots  < \vartheta_j^i < \ldots < \vartheta_i^i < \vartheta_{i+1}^i < \theta
\label{eqn_thm_localiss_proof_varthetaji_ineq}
\end{align}
Then, consider the trajectories of systems with indices $j\leq i$ in the time interval $[\vartheta_{j+1}^i,\theta]$ by applying (\ref{eqn_thm_localiss}), hereby using the bound on its initial condition at $\vartheta_j$. This yields the bound
\begin{align}
\|x_j\|_{\infty}^{[\vartheta_{j+1}^i,\theta]} &\leq \beta\big(|x_j(\vartheta_j^i)|,\vartheta_{j+1}^i-\vartheta_j^i\big)
+ \bar{\gamma}\|x_{j-1}\|_{\infty}^{[\vartheta_j^i,\theta]} \nonumber\\
&\qquad+ \sigma\!\left(\|w_j\|_{\infty}^{[\vartheta_j^i,\theta]}\right)
\label{eqn_thm_localiss_proof_step6}
\end{align}
for all $j\leq i\in\Ical_N$, whereas the bound for $j = 0$ reads
\begin{align}
\|x_0\|_{\infty}^{[\vartheta_{1}^i,\theta]} \leq \beta\big(|x_0(\vartheta_0^i)|,\vartheta_1^i-\vartheta_0^i\big) + \sigma\!\left(\|w_0\|_{\infty}^{[\vartheta_0^i,\theta]}\right).
\label{eqn_thm_localiss_proof_step7}
\end{align}
Similar to before, the recursive application of (\ref{eqn_thm_localiss_proof_step6}) and the use of (\ref{eqn_thm_localiss_proof_step7}) can be shown to lead to
\begin{align}
\|x_i\|_{\infty}^{[\vartheta_{i+1}^i,\theta]} &\leq
\sum_{j=0}^i \bar{\gamma}^{i-j} \beta\big(|x_j(\vartheta_j^i)|,\vartheta_{j+1}^i-\vartheta_j^i\big) \nonumber\\
&\qquad+ \sum_{j=0}^i \bar{\gamma}^{i-j} \sigma\!\left(\|w_j\|_{\infty}^{[\vartheta_j^i,\theta]}\right).
\label{eqn_thm_localiss_proof_step8}
\end{align}
Here, it is recalled that the choice of the parameters $\bar{c}$ and $\bar{c}_w$ guarantees that $\|x_i\|_{\infty} < c$, enabling the repeated application of (\ref{eqn_thm_localiss}).

In order to show that the first term on the right-hand-side of (\ref{eqn_thm_localiss_proof_step8}) can be bounded by a function of class $\classKL$, the sequence $\{\vartheta_j^i\}_{j=0}^{i+1}$ is chosen as
\begin{align}
\vartheta_{j}^i = \theta - (1 - \bar{\omega})\sum_{l=0}^{1+i-j}\bar{\omega}^l (\theta - \theta_0),
\label{eqn_thm_localiss_proof_varthetaji_def}
\end{align}
such that
\begin{align}
\vartheta_{j+1}^i - \vartheta_{j}^i = (1 - \bar{\omega})\bar{\omega}^{1+i-j}(\theta - \theta_0)
\label{eqn_thm_localiss_proof_dvarthetaji}
\end{align}
for any $0 \leq j\leq i$. Here, $0<\bar{\omega}<1$ is a parameter that will be specified later. By this choice, (\ref{eqn_thm_localiss_proof_varthetaji_def}) represents a geometric series in which the time intervals (\ref{eqn_thm_localiss_proof_dvarthetaji}) shrink as subsystems further away from system $i$ are considered. Moreover, it is clear by the scaling with $1-\bar{\omega}$ that (\ref{eqn_thm_localiss_proof_varthetaji_ineq}) holds for any $i\in\Ical_N$ and $N\in\N$.

Next, define a function $\phi$ as
\begin{align}
\phi(r,s) \defl \sup_{\omega\in(0,1]} \omega^q \beta(r,\omega s)
\label{eqn_thm_localiss_proof_phi}
\end{align}
for some $q>0$. From this definition it follows that $\phi$ is of class $\classKL$ and that $\phi(r,s) \geq \beta(r,s)$ for all $r,s\geq0$, where equality holds if $\beta$ satisfies the condition (\ref{eqn_thm_localiss_beta}). In fact, $\phi$ in (\ref{eqn_thm_localiss_proof_phi}) always satisfies the condition (\ref{eqn_thm_localiss_beta}). Namely, for any $\tilde{\omega}$ such that $0<\tilde{\omega}\leq1$, it follows from (\ref{eqn_thm_localiss_proof_phi}) that
\begin{align}
\tilde{\omega}^q\phi(r,\tilde{\omega} s)
&= \sup_{\omega\in(0,1]} \tilde{\omega}^q\omega^q\beta(r,\tilde{\omega}\omega s), \nonumber\\
&= \sup_{c\in(0,\tilde{\omega}]} c^q\beta(r,cs), \nonumber\\
&\leq \sup_{c\in(0,1]} c^q\beta(r,cs)
= \phi(r,s).
\label{eqn_thm_localiss_proof_phi_ineq}
\end{align}
Using the fact that $\beta(r,s)\leq\phi(r,s)$ for all $r,s\geq0$ and the choice of the intervals (\ref{eqn_thm_localiss_proof_dvarthetaji}), the first term on the right-hand-side of (\ref{eqn_thm_localiss_proof_step8}) can be bounded as
\begin{align}
\sum_{j=0}^i &\bar{\gamma}^{i-j} \beta\big(|x_j(\vartheta_j^i)|,\vartheta_{j+1}^i-\vartheta_j^i\big) \nonumber\\
&\leq \sum_{j=0}^i \bar{\gamma}^{i-j} \phi\Big(|x_j(\vartheta_j^i)|,(1-\bar{\omega})\bar{\omega}^{1+i-j}(\theta - \theta_0) \Big),
\label{eqn_thm_localiss_proof_step9}\\
&\leq \sum_{j=0}^i \frac{\bar{\gamma}^{i-j}}{\big((1-\bar{\omega})\bar{\omega}^{1+i-j}\big)^q}
\phi\big(|x_j(\vartheta_j^i)|,\theta - \theta_0 \big),
\label{eqn_thm_localiss_proof_step10}\\
&= \sum_{j=0}^i \frac{1}{(1 - \bar{\omega})^q\bar{\omega}^q}
\bigg(\frac{\bar{\gamma}}{\bar{\omega}^q}\bigg)^{\!i-j} \phi\big(|x_j(\vartheta_j^i)|,\theta - \theta_0 \big),\label{eqn_thm_localiss_proof_step11}
\end{align}
where the property (\ref{eqn_thm_localiss_proof_phi_ineq}) is used to obtain (\ref{eqn_thm_localiss_proof_step10}). Here, it is noted that $0<(1-\bar{\omega})\bar{\omega}^{1+i-j}<1$ for any $0\leq j\leq i$ as $\bar{\omega}$ satisfies $0<\bar{\omega}<1$, such that (\ref{eqn_thm_localiss_proof_phi_ineq}) can indeed be applied. Even though (\ref{eqn_thm_localiss_proof_step11}) provides a time-dependent upper bound, it is not yet of the form (\ref{eqn_def_diststringstab}) due to the appearance of the norm $|x_j(\vartheta_j^i)|$. Therefore, it is recalled that this norm can be bounded through (\ref{eqn_thm_localiss_proof_step5}), which, by using the notation (\ref{eqn_thm_localiss_proof_delta}), (\ref{eqn_thm_localiss_proof_Delta}), gives
\begin{align}
|x_j(\vartheta_j^i)| \leq \infnorm{x_j}{[\theta_0,\theta]} \leq \delta\!\left(\sup_{k\in\Ical_N^0}|x_k(0)|\right) + \Delta^{[\theta_0,\theta]}.
\label{eqn_thm_localiss_proof_step12}
\end{align}
Next, it is remarked that the parameter $\bar{\omega}$ satisfying $0<\bar{\omega}<1$ can be chosen such that $\bar{\gamma}<\bar{\omega}^q<1$, which follows from the property $0<\bar{\gamma}<1$. Using this choice, the function $\tilde{\beta}$ defined as
\begin{align}
\tilde{\beta}(r,s) \defl \frac{1}{(1 - \bar{\omega})^q}\frac{1}{\bar{\omega}^q - \bar{\gamma}}\phi(r,s),
\label{eqn_thm_localiss_proof_step13}
\end{align}
is well-defined, of class $\classKL$, and satisfies
\begin{align}
\sum_{j=0}^i \frac{1}{(1 - \bar{\omega})^q\bar{\omega}^q} \bigg(\frac{\bar{\gamma}}{\bar{\omega}^q}\bigg)^{\!i-j} \phi\big(r,s \big) \leq \tilde{\beta}(r,s),
\label{eqn_thm_localiss_proof_step14}
\end{align}
where the inequality follows by noting that the sum at the left-hand-side of (\ref{eqn_thm_localiss_proof_step14}) represents a convergent series due to $0<\bar{\gamma}<\bar{\omega}^q$. Now, after the substitution of (\ref{eqn_thm_localiss_proof_step12}) in (\ref{eqn_thm_localiss_proof_step11}) and the use of the upper bound (\ref{eqn_thm_localiss_proof_step14}), as well as the observation that $|x_i(\theta)|\leq \infnorm{x_i}{[\vartheta_{i+1}^i,\theta]}$, it follows that (\ref{eqn_thm_localiss_proof_step8}) can be bounded as
\begin{align}
|x_i(\theta)| &\leq \tilde{\beta}\!\left( \delta\!\left(\sup_{j\in\Ical_N^0}|x_j(\theta_0)|\right) + \Delta^{[\theta_0,\theta]}, \theta - \theta_0\right) \nonumber\\
&\qquad+ \Delta^{[\theta_0,\theta]},
\label{eqn_thm_localiss_proof_step15}
\end{align}
where the bound on the disturbance-dependent terms in (\ref{eqn_thm_localiss_proof_step8}) follows from (\ref{eqn_thm_localiss_proof_Delta}). It is noted that the bound (\ref{eqn_thm_localiss_proof_step15}) holds for any $i\in\Ical_N^0$ and all $N\in\N$ and thus presents a uniform bound as in the definition of disturbance string stability.

In order to address the appearance of $\Delta^{[\theta_0,\theta]}$ in the argument of the class $\classKL$ function $\tilde{\beta}$ in (\ref{eqn_thm_localiss_proof_step15}), it is recalled that $|x_i(\theta)|$ also satisfies the bound (\ref{eqn_thm_localiss_proof_step5}). It is therefore natural to consider the tightest of the bounds (\ref{eqn_thm_localiss_proof_step5}) and (\ref{eqn_thm_localiss_proof_step15}) through the introduction of the function
\begin{align}
\kappa(r,\Delta,\theta-\theta_0) \defl \min\!\big\{ \tilde{\beta}\big(\delta(r) + \Delta,\theta-\theta_0\big), \delta(r) \big\},
\label{eqn_thm_localiss_proof_kappa}
\end{align}
where the subscript in $\Delta^{[\theta_0,\theta]}$ is omitted for ease of exposition.
In particular, the function $\kappa$ satisfies
\begin{align}
\kappa(r,\Delta,\theta-\theta_0) &\leq \kappa\big(r,\alpha^{-1}(r),\theta-\theta_0 \big) \nonumber\\
&\qquad+ \kappa\big(\alpha(\Delta),\Delta,\theta-\theta_0\big)
\label{eqn_thm_localiss_proof_kappaineq}
\end{align}
for any function $\alpha$ of class $\classKinf$ (see \cite{jiang_1994}). By selecting any of the two terms in the minimum in the definition of $\kappa$ in (\ref{eqn_thm_localiss_proof_kappa}), the inequality (\ref{eqn_thm_localiss_proof_kappaineq}) leads to
\begin{align}
\!\!\!\kappa(r,\Delta,\theta-\theta_0) \leq \tilde{\beta}\big(\delta(r)+\alpha^{-1}(r),\theta-\theta_0\big) + \delta\circ\alpha(\Delta),\!\!\!
\label{eqn_thm_localiss_proof_step16}
\end{align}
such that
\begin{align}
|x_i(\theta)| &\leq \bar{\beta}\!\left(\sup_{i\in\Ical_N^0} |x_i(\theta_0)|,\theta-\theta_0\right) \nonumber\\
&\qquad+ \bar{\sigma}\!\left(\sup_{i\in\Ical_N^0} \|w_i\|_{\infty}^{[\theta_0,\theta]}\!\right).
\label{eqn_thm_localiss_proof_step17}
\end{align}
Here, the function $\bar{\beta}$ is defined as $\bar{\beta}(r,\vartheta) \defl \tilde{\beta}(\delta(r)+\alpha^{-1}(r),\vartheta)$ with $\tilde{\beta}$ as obtained through (\ref{eqn_thm_localiss_proof_step13}) and (\ref{eqn_thm_localiss_proof_phi}). As a result, $\bar{\beta}$ is of class $\classKL$. Moreover, by using the definition of $\Delta = \Delta^{[\theta_0,\theta]}$ in (\ref{eqn_thm_localiss_proof_Delta}) it follows that $\bar{\sigma}$ is given by $\bar{\sigma}(r) = (\id + \delta\circ\alpha)((1-\bar{\gamma})^{-1}\sigma(r))$ with $\delta$ as in (\ref{eqn_thm_localiss_proof_delta}) and where $\id$ denotes the identity function satisfying $\id(r) = r$ for all $r\geq0$. Then, $\bar{\sigma}$ is of class~$\classKinf$.

It is recalled that the bound (\ref{eqn_thm_localiss_proof_step17}) applies to any initial condition satisfying $|x_i(\theta_0)| < \bar{c}$ and disturbance $w_i$ satisfying $\|w_i\|_{\infty} < \bar{c}_w$ and holds for all $i\in\Ical_N^0$ and all $N\in\N$. As a result, the first statement in Theorem~\ref{thm_localiss} is proven.
The second statement follows by noting that (\ref{eqn_thm_localiss_beta}) implies $\phi = \beta$ in (\ref{eqn_thm_localiss_proof_phi}) and the definition (\ref{eqn_thm_localiss_proof_step13}).
Finally, it can easily be observed that the results obtained in this proof hold globally when $c = \infty$ and $c_w = \infty$, proving the third statement.

\subsection{Proof of Theorem~\ref{thm_localios}}\label{app_thm_localios_proof}
The proof of this theorem will rely on the ideas developed in the proof of Theorem~\ref{thm_localiss}.

Thereto, constants $\bar{c}$ and $\bar{c}_w$ are introduced satisfying $0<\bar{c}<c$ and $0<\bar{c}_w<c_w$. Then, the recursive application of (\ref{eqn_thm_localios}) for $\theta = \theta_0$, hereby taking initial conditions $|x_i(\theta)| < \bar{c}$ and disturbances $\|w_i\|<\bar{c}_w$, yields
\begin{align}
\|y_i\|_{\infty} &\leq \frac{1}{1-\bar{\gamma}} \, \beta_y\!\!\left( \sup_{j\in\Ical_N^0}|x_j(\theta_0)|, 0\right) \nonumber\\
&\qquad+ \frac{1}{1-\bar{\gamma}} \, \sigma_y\!\!\left(\sup_{j\in\Ical_N^0}\|w_j\|_{\infty}\right),
\label{eqn_thm_localios_proof_step1}
\end{align}
for all $i\in\Ical_N^0$ and $N\in\N$, analogous to (\ref{eqn_thm_localiss_proof_step5}) in the proof of Theorem~\ref{thm_localiss}. Then, the substitution of (\ref{eqn_thm_localios_proof_step1}) in (\ref{eqn_thm_localios_iss}) leads to a boundedness of trajectories $x_i$ as
\begin{align}
\|x_i\|_{\infty} \leq \delta\!\left( \sup_{j\in\Ical_N^0}|x_j(\theta_0)| \right) + \Delta\!\left(\sup_{j\in\Ical_N^0}\|w_j\|_{\infty}\right),
\label{eqn_thm_localios_proof_step2}
\end{align}
where the functions $\delta$ and $\Delta$ of class $\classKinf$ are given as
\begin{align}
\delta(r) &\defl \beta_x(r,0) + \gamma_x\!\left(\frac{2}{1-\bar{\gamma}}\beta_y(r,0)\right),
\label{eqn_thm_localios_proof_delta}\\
\Delta(r) &\defl \sigma_x(r) + \gamma_x\!\left(\frac{2}{1-\bar{\gamma}} \,\sigma_y(r)\right).
\label{eqn_thm_localios_proof_Delta}
\end{align}
In the derivation of (\ref{eqn_thm_localios_proof_delta}) and (\ref{eqn_thm_localios_proof_Delta}), the property $\gamma_x(r_1 + r_2)\leq\gamma_x(2r_1) + \gamma_x(2r_2)$ is used. Now, choosing the constants $\bar{c}<c$ and $\bar{c}_w<c_w$ to satisfy $\delta(\bar{c})+\Delta(\bar{c}_w) < c$ ensures that conditions in the statement of the theorem hold.

In order to show that the effect of initial condition vanishes as $\theta\rightarrow\infty$, the ideas in the proof of Theorem~\ref{thm_localiss} are adopted. Namely, analogous to (\ref{eqn_thm_localiss_proof_step15}) in Appendix~\ref{app_thm_localiss_proof}, there exists a function $\tilde{\beta}_y$ of class $\classKL$ such that
\begin{align}
|y_i(\theta)| &\leq \tilde{\beta}_y\!\left( \delta\!\left( \sup_{j\in\Ical_N^0}|x_j(\theta_0)| \!\right) \!+\! \Delta\!\left(\sup_{j\in\Ical_N^0}\|w_j\|_{\infty}\!\right)\!, \theta\!-\!\theta_0 \right) \nonumber\\
&\qquad+ \frac{1}{1-\bar{\gamma}} \, \sigma_y\!\left( \sup_{j\in\Ical_N^0}\|w_j\|_{\infty}\right).
\label{eqn_thm_localios_proof_step3}
\end{align}
Here, (\ref{eqn_thm_localios_proof_step2}) is used to bound estimates of the initial conditions. As in the proof of Theorem~\ref{thm_localiss}, (\ref{eqn_thm_localios_proof_step3}) is a uniform bound and it holds for all $i\in\Ical_N^0$ and all~$N\in\N$.

Next, in order to combine the bounds (\ref{eqn_thm_localios_proof_step1}) and (\ref{eqn_thm_localios_proof_step3}), the function $\kappa$ is introduced as
\begin{align}
\kappa(r,\Delta(s),\theta-\theta_0) &\defl \min\!\big\{ \tilde{\beta}_y\big(\delta(r) + \Delta(s),\theta-\theta_0\big),\nonumber\\
&\qquad\hspace{2.5cm} \tfrac{1}{1-\bar{\gamma}}\beta_y(r,0) \big\},
\end{align}
which is of the same form as (\ref{eqn_thm_localiss_proof_kappa}). Consequently, a bound of the form (\ref{eqn_thm_localiss_proof_kappaineq}) holds, which allows for obtaining a bound of the form
\begin{align}
|y_i(\theta)| &\leq \bar{\beta}_y\!\!\left( \sup_{i\in\Ical_N^0}|x_i(\theta_0)|, \theta-\theta_0 \right) \nonumber\\
&\qquad+ \bar{\sigma}_y\!\!\left( \sup_{i\in\Ical_N^0}\|w_i\|_{\infty}^{[\theta_0,\theta]} \right).
\label{eqn_thm_localios_proof_step4}
\end{align}
Here, $\bar{\beta}_y(r,s)\defl\tilde{\beta}_y(\delta(r)+\alpha^{-1}(r),\vartheta)$ and $\bar{\sigma}_y(r) \defl (1-\bar{\gamma})^{-1}( \beta_y(\Delta(r),0) + \sigma_y(r) )$, with $\delta$ and $\Delta$ as in (\ref{eqn_thm_localios_proof_delta}) and (\ref{eqn_thm_localios_proof_Delta}), respectively.

Now, the proof can be finalized by noting that the substitution of (\ref{eqn_thm_localios_proof_step4}) in (\ref{eqn_thm_localios_iss}) leads to a bound of the form (\ref{eqn_def_diststringstab}) through the use of standard results on the cascade interconnection of input-to-state stable systems (see, e.g., \cite{sontag_1989,book_krstic_1995}).

\subsection{Proof of Lemma~\ref{lem_control}}\label{app_lem_control_proof}
In order to prove the lemma, it is first noted that $|x|_{\Scal_i} = |\delta_i|$. Then, after introducing the function $V(x) = \delta_i^{\T}P\delta_i$ for some $P = P^{\T}\succ0$, it is clear that $\alpha_1(|x|_{\Scal_i}) \leq V(x) \leq \alpha_2(|x|_{\Scal_i})$ for some functions $\alpha_1$, $\alpha_2$ of class $\classKinf$.

By asymptotic stability of the matrix $A+\kappa BK$, it follows that $P$ can be chosen to satisfy
\begin{align}
(A+ \kappa BK)^{\T}P + P(A+\kappa BK) \prec -I,
\end{align}
where it is noted that controllability of the pair $(A,B)$ (see (\ref{eqn_platoonsys_space_matrices})) ensures that an asymptotically stabilizing feedback matrix $K$ exist. Then, after substituting $\tilde{u}_i = K\delta_i$ in (\ref{eqn_platoonsys_space_x}), the space differentiation of $V$ along trajectories of the resulting controlled platoon system yields
\begin{align}
\mathring{V}(x) &\leq -|\delta_i|^2 + 2\delta_i^{\T}P\bar{\rho}(\xi_i,\xi_{i-1},\xi_0)\bar{w}_i,\\
&\leq -|\delta_i|^2 + 2|\delta_i| \|P\| \|\bar{\rho}(\xi_i,\xi_{i-1},\xi_0)\| |\bar{w}_i|.
\label{eqn_lem_control_proof_step1}
\end{align}
At this point, it is noted that $x\in\Xcal_{\bar{c}_{\delta}}$ for some $\bar{c}_{\delta}>0$ implies that the velocity tracking errors $e_{1,i}$ are bounded for all $i\in\Ical_N^0$, as follows from their relation to the state in (\ref{eqn_delta1}). In fact, there exists $\bar{c}_{\delta}$ such that the velocities $v_i$ satisfy $v_i > 0$ for all $x\in\Xcal_{\bar{c}_{\delta}}$, as follows from (\ref{eqn_veltrackerror_e1}) and Assumption~\ref{ass_vref}. Then, the functions in (\ref{eqn_platoonsys_space_functions}) are well-defined and smooth, which implies by the definition (\ref{eqn_rhok}) that $\bar{\rho}$ in (\ref{eqn_rhobar1})--(\ref{eqn_rhobark}) is smooth. As this function is evaluated on the compact set $\Xcal_{\bar{c}_{\delta}}$, it follows that there exists a constant $c_{\rho}>0$ such that $\|\bar{\rho}(\xi_i,\xi_{i-1},\xi_0)\| < c_{\rho}$. The substitution of this bound in (\ref{eqn_lem_control_proof_step1}) yields, for any $\alpha$ satisfying $0<\alpha<1$,
\begin{align}
\mathring{V}(x) \leq -\alpha|\delta_i|^2 - |\delta_i| \big( (1-\alpha)|\delta_i| - 2c_{\rho}\|P\| |\bar{w}_i| \big),
\end{align}
which leads to the implication
\begin{align}
|x|_{\Scal_i} &= |\delta_i| \geq \frac{2c_{\rho}\|P\|}{1-\alpha}|\bar{w}_i| \nonumber\\
&\;\;\implies\;\;
\mathring{V}(x) \leq -\alpha|\delta_i|^2 = -\alpha|x|_{\Scal_i}^2.
\label{eqn_lem_control_proof_step2}
\end{align}
Following \cite{lin_1995}, (\ref{eqn_lem_control_proof_step2}) implies the result (\ref{eqn_lem_control_setiss}), finalizing the proof of this lemma.

\subsection{Proof of Theorem~\ref{thm_diststringstab}.}\label{app_thm_diststringstab_proof}
In order to prove the theorem, it will first be shown that any controller that achieves (\ref{eqn_lem_control_setiss}) guarantees disturbance string stability. Then, a constant $\bar{c}_{\delta}$ and the set $\Xcal_{\bar{c}_{\delta}}$ will be considered for which the feedback controller $\tilde{u}_i = K\delta_i$ achieves (\ref{eqn_lem_control_setiss}) through Lemma~\ref{lem_control}. In this case, it will be shown that there exist a set of initial conditions and set of disturbances that ensure that $\Xcal_{\bar{c}_{\delta}}$ is invariant, thus satisfying the conditions of Definition~\ref{def_diststringstab}.

In order to obtain a tight upper bound on the input-to-output gain of (\ref{eqn_platoonsys_space_x}) with input $y_{i-1}$ and output $y_i$, the solution of the dynamics for $\Delta_i$ in (\ref{eqn_platoonsys_space_x}) is written explicitly in order to obtain
\begin{align}
|\Delta_i(s)| &\leq \big\|e^{-\kappa^{-1}(s-s_0)}\big\||\Delta_i(s_0)| \nonumber\\
&\qquad+ \int_{s_0}^s\big\|\kappa^{-1}e^{-\kappa^{-1}(s-\vartheta)}\big\| |\delta_{1,i}(\vartheta)| \di\vartheta \nonumber\\
&\qquad+\|y_{i-1}\|_\infty^{[s_0,s]},
\label{eqn_thm_diststringstab_proof_step1}
\end{align}
where the final term is obtained by using
\begin{align}
\int_{s_0}^s\big\|\kappa^{-1}e^{-\kappa^{-1}(s-\vartheta)}\big\| |y_{i-1}(\vartheta)| \di\vartheta
\leq \|y_{i-1}\|_\infty^{[s_0,s]}.
\end{align}
Then, by noting that $|\delta_{1,i}|\leq|\delta_i|=|x|_{\Scal_i}$, it can be observed that the use of a controller that satisfies (\ref{eqn_lem_control_setiss}) leads to a bound on $\Delta_i$ of the form
\begin{align}
|\Delta_i(s)| &\leq \beta_{\Delta}\big(|x_i(s_0)|, s-s_0\big) + \|y_{i-1}\|_{\infty}^{[s_0,s]} \nonumber\\
&\qquad+ \sigma_{\Delta}\!\left(\|\bar{w}_i\|_{\infty}^{[s_0,s]}\right),
\label{eqn_thm_diststringstab_proof_step2}
\end{align}
for some functions $\beta_{\Delta}$ of class $\classKL$ and $\sigma_{\Delta}$ of class $\classKinf$. Next, the output equation in (\ref{eqn_platoonsys_space_x}) yields
\begin{align}
|y_i| \leq (1-\kappa_0)||\Delta_i(s)| + |\delta_i(s)|
\label{eqn_thm_diststringstab_proof_step3}
\end{align}
such that the substitution of the bounds for $\Delta_i$ in (\ref{eqn_thm_diststringstab_proof_step2}) and $\delta_i$ in (\ref{eqn_lem_control_setiss}) implies input-to-output stability of (\ref{eqn_platoonsys_space_x}) subject to any controller that satisfies (\ref{eqn_lem_control_setiss}). Specifically, the disturbances $\bar{w}_i$ act as inputs and a bound of the form (\ref{eqn_thm_localios}) holds with $\gamma_y(r) = (1-\kappa_0)r$. Similarly, by noting that $|x_i| = |\Delta_i| + |\delta_i|$, input-to-state stability of (\ref{eqn_platoonsys_space_x}) follows and a bound of the form (\ref{eqn_thm_localios_iss}) holds. Then, by Theorem~\ref{thm_localios}, the platoon given in (\ref{eqn_platoonsys_space_x}) with a controller satisfying (\ref{eqn_lem_control_setiss}) is disturbance string stable.

Next, consider the specific feedback controller $\tilde{u}_i = K\delta_i$ as in the statement of Lemma~\ref{lem_control}. By this lemma, there exists a constant $\bar{c}_{\delta}$ such that the controller achieves input-to-state stability with respect to the set $\Scal_i$ as in (\ref{eqn_lem_control_setiss}) for trajectories satisfying $x(s)\in\Xcal_{\bar{c}_{\delta}}$ for all $s\geq0$. Consequently, the results on disturbance string stability in Theorem~\ref{thm_localios} hold for these trajectories, as shown in the first part of this proof. However, it is noted that the constants $c$ ($c<\bar{c}_{\delta}$) and $c_w$ in the statement of Theorem~\ref{thm_localios} can be chosen such that the constants $\bar{c}$ and $\bar{c}_w$ in (\ref{eqn_def_diststringstab_bounds}) in the definition of disturbance string stability satisfy $\bar{\beta}(\bar{c},0) + \bar{\sigma}(\bar{c}_w) < \bar{c}_{\delta}$. In this case, the set $\Xcal_{\bar{c}_{\delta}}$ is invariant and the conditions in Lemma~\ref{lem_control}, which are required for controller design, indeed hold. As a result, the feedback controller $\tilde{u}_i = K\delta_i$ (when implemented for all $i\in\Ical_N^0$) achieves disturbance string stability, proving the theorem.

\bibliographystyle{plain}
\bibliography{references}

\end{document}